\documentclass[12pt,final]{elsarticle}

\newcommand\GGG{\color{blue}}

\newcommand\BBB{\color{black}}

\usepackage{cancel}
\usepackage{mathtools}
\mathtoolsset{showonlyrefs}
\usepackage[colorinlistoftodos,textsize=footnotesize,textwidth=1.4\marginparwidth]{todonotes}
\usepackage{autobreak}
\allowdisplaybreaks
\usepackage{amsmath}
\usepackage{enotez}
\usepackage{fullpage}
\usepackage{hyperref}

\usepackage{amssymb,amsmath,amsthm,float}

\usepackage[inner]{showlabels}

\usepackage{xcolor}

\usepackage{amsmath}
\usepackage{import}
\usepackage{color}

\usepackage{bm}

\newcommand {\bydef}{\,\raise.07485ex\hbox{:}\kern-.025em\hbox{=}\,}

\newcommand{\R}{\mathbb{R}}

\theoremstyle{definition}
\newtheorem{lemma}{Lemma}

\newtheorem{proposition}{Proposition}
\newtheorem{theorem}{Theorem}

\newtheorem{remark}{Remark}
\newtheorem{definition}{Definition}

\begin{document}


\begin{frontmatter}

\title{A theory of magneto-elastic nanorods obtained through rigorous dimension reduction}

\author[sapienza]{Jacopo Ciambella}
\ead{jacopo.ciambella@uniroma1.it}
\affiliation[sapienza]{organization={Dipartimento di Ingegneria Strutturale e Geotecnica, Sapienza Universit\`a di Roma},
            addressline={Via Eudossiana 18}, 
            city={Roma},
            postcode={00184}, 
            country={Italy}}
\author[utia]{Martin Kru\v z\'ik}
\ead{kruzik@utia.cas.cz}
\affiliation[utia]{organization={Czech Academy of Sciences, Institute of Information Theory and Automation},
          addressline={Pod Vod\'arenskou v\v e\v z\'i 4}, 
         city={Prague 8},
    postcode={ CZ-182 00}, 
       country={Czechia}}
\author[romatre]{Giuseppe Tomassetti}
\ead{giuseppe.tomassetti@uniroma3.it}
\affiliation[romatre]{organization={Engineering Department, Roma Tre University},
          addressline={Via Vito Volterra 62}, 
         city={Roma},
    postcode={00146}, 
       country={Italy}}

\begin{abstract}
  Starting from a two-dimensional theory of magneto-elasticity for fiber-reinforced magnetic elastomers we carry out a rigorous dimension reduction to derive a rod model that describes a thin magneto-elastic strip undergoing planar deformations. The main features of the theory are the following: a magneto-elastic interaction energy that manifests itself through a distributed torque; a penalization term that prevent local interpenetration of matter; a regularization that depends on the second gradient of the deformation and models microstructure-induced size effects. As an application, we study a problem involving magnetically-induced buckling and we show that the intensity of the field at the onset of the instability increases if the length of the rod is decreased. Finally, we assess the accuracy of the deduced model by performing numerical simulations where we compare the two-dimensional and the one-dimensional theory in some special cases and we observe excellent agreement.
\end{abstract}

\end{frontmatter}



\section{Introduction}

Magnetorheological elastomers (MREs) represent  a class of soft composites in which magnetic fillers are embedded into a soft polymeric matrix. The combined use of a compliant matrix and rigid particles make MREs mechanically sensitive to applied magnetic fields and  allows for their  static and dynamic shape programming as well as reconfiguration capabilities, which have been exploited in a number of engineering applications \cite{Wu2020}. These functionalities are particularly appealing when the overall dimensions of the actuators are at the millimeter scale, or even smaller, since at those scales, magnetic fields can be specified not only in magnitude, but also in direction and spatial gradients, enlarging the class of achievable shapes  \cite{Lum2016}. Over the past decades, many applications have been proposed to exploit MREs as gripping tools \cite{Ze2020}, manipulators \cite{Xu2019}, micro-swimmers \cite{Hu2018,Chen2020} or in biomimetic applications \cite{Evans2007,Hanasoge2020}.

A key step of the manufacturing process is the ``curing'' process that  creates long-chain polymers connected by cross--links  responsible for the elastic nature and recovery characteristics of the finished material. In addition, since curing is usually carried out under a magnetic field, the magnetic particles align themselves and form chain-like microstructures which make the finished material transversely isotropic \cite{Ciambella2018}. The presence of these aligned chains produce an actuation mechanism in the form of distributed torques which yields better steering and navigational capabilities at much smaller scales, and differentiate them from previously developed continuum robots \cite{Rikken2014,Stanier2016}.

To account for this actuation mechanism, several models were proposed in the literature. The theory of bulk magneto--elasticity was studied in the 1950s  with the pioneering works of Brown \cite{brown1966}, Truesdell, Toupin and Tiersten \cite{Tiersten1964} and later formulated in the context of large strain elasticity by Dorfmann and Ogden \cite{Dorfmann2003} and Kankanala and Triantafyllidis \cite{Kankanala2004}. The advantage of structural theories with magnetic forcing terms lies in their analytical tractability which allows, for instance, to study closed form solution for buckling and post-buckling regimes as well as inverse problems related to shape optimization \cite{Ciambella2020,durastanti,Lum2016,sano2021}.

Recent works have shown the feasibility of the so-called \emph{fiber-reinforced magnetic elastomers}. These materials are obtained by replacing the standard spherical fillers with carbon fibers coated with nickel \cite{Ciambella2017,Stanier2016}. The accurate control of the fiber orientation during the production process of the elastomer could have a significant effect on the roughness properties of the surface and potentially lead to new engineering applications including dampers and actuators. A continuum model for this new type of magnetic elastomers has been developed in \cite{Ciambella2018}.

In this paper we use as starting point the model proposed in \cite{Ciambella2018} to construct, through a rigorous deduction, a one--dimensional model for magneto-elastic nanorods. Being interested in actuators undergoing planar deformations, we consider the special case of plane elasticity. Moreover, in order to account for the internal length scale brought about by the embedded fibers we include a regularization that depends on the second gradient of the deformation (see the energy functional $\mathcal E$ defined in \eqref{eq:3} in Section \ref{sec:magneto_elastic_energy}). This additional term accounts for the microstructure related size effects which become significant when the specimen is very small \cite{Murmu2014,schraad1997}, and provides an alternative to non-local approaches based on integral-type formulations. By incorporating size effects, the resulting continuum theory is able to match the results of molecular dynamics simulations \cite{Duan2007,Li2015}. Thanks to this regularization we are able to apply the direct method of the calculus of variations to prove the existence of a minimizer for the equilibrium problem for a body of general shape in a fairly wide class of loading environments.  We then focus our attention on a model problem featuring a rod-like body having the shape of strip of constant length and vanishing thickness, clamped on one side and free at the other. Using techniques mutuated from $\Gamma$-convergence, we prove that the solutions of the equilibrium problem converge to the minimizers of an energy functional (see \eqref{eq:32} of Section \ref{sec:convergence_result}) whose arguments are two vector fields defined on the axis of the strip. These field describe placement, rotation, and stretch of the typical cross section of the rod. We derive the equations that govern the equilibrium of the rod, we perform a linearization, and we apply them to the problem of the buckling of a rod induced by a magnetic field, which we deal with through a semi-analytical approach to show how the presence of the internal scale induces a size-dependent effect. Finally, in order to assess the aforementioned convergence of solutions, we offer a numerical comparison between the two-dimensional and the one-dimensional theories.

The structure of the paper is as follows. In Sect. 2 the main modeling assumptions are introduced and the magneto--elastic energy of the composite material presented. The existence of minimizers of the energy is proved in Sect. 3, whereas the reduced order energy, obtained by $\Gamma$-convergence, is deduced in Sect. 4. The equations governing the motion of a clamped-free magneto--elastic rod are discussed in Sect. 5, whereas their linearization is presented in Sect. 6, and applied in Sect. 7 to determine the size-dependence of the critical field that induces magneto-elastic buckling. Finally in Sect. 8 the instability of the magnetic rod in a cantilever configuration is studied and the corresponding buckling load derived either numerically and analytically.

\section{The magneto-elastic energy}\label{sec:magneto_elastic_energy}
We use the standard notation for Lebesgue, Sobolev, and H\"older spaces (see for example \cite{Fucik-Kufner-John}). We consider a domain $\Omega\subset\mathbb R^2$ representing the reference configuration of the magneto-elastic body. We assume that the possible configurations of the body take place in the plane $\mathbb R^2$, and we represent each configuration through a deformation $y:\Omega\to\mathbb R^2$.

The orientation of the magnetic fibers in the reference configuration is described by a vector field $a:\Omega\to\mathbb R^2$  such that $|a|=1$. These fibers confer to the material both mechanical and magnetic anisotropy. As to the mechanical response, we describe it by a strain-energy density of the form
 \begin{equation}\label{eq:9}
    W(x,F)=W_e(F,a(x)\otimes a(x)),
  \end{equation}
  where $F=\nabla y$ is the deformation gradient, $W_e:\mathbb R^{2\times 2}\times \mathbb R^{2\times 2}\to(-\infty;+\infty]$ is a continuous isotropic function such that  $W_e(\cdot,A)$ is frame indifferent for every $A\in\mathbb R^{2\times 2}$. Thus, given any $F\in\mathbb R^{2\times 2}$ and any $A\in\mathbb R^{2\times 2}$, we have $W(Q^TFQ,Q^TAQ)=W_e(F,A)$ (isotropic response)  and $W(Q^TF,A)=W_e(F,A)$ (frame indifference) for every rotation matrix $Q$. It is worth pointing out that while the first argument of $W_e$ is subjected to the requirement of frame indifference, the second argument is a material parameter, fixed once and for all.

We denote by $h:\mathbb R^2\to\mathbb R^2$ the applied magnetic field permeating the entire space and we model the interaction between the body and the applied field through the following energy density, i.e., \emph{Zeeman} energy, $\mathfrak M:\Omega\times\R^2\times\R^{2\times 2}\to\R$ 
   \begin{equation}\label{eq:53}
     \mathfrak M(x,y,F)=-M(x)\color{black} h(y)\cdot \frac{F a(x)}{|F a(x)|},
  \end{equation}
  where $M:\Omega\to\R$ represents the  magnitude  of the permanent magnetization. For a detailed derivation of \eqref{eq:53} the interested reader is referred to \cite{Ciambella2018}. The main assumptions underlying \eqref{eq:53} are: (i) the material is composed of a dilute suspension of particles magnetically and elastically anisotropic according to the average direction $a$ at each material point, (ii) being dilute, the interaction between particles can be neglected, (iii) the spatial variation of the externally applied magnetic field can be neglected at the particle scale, meaning that the resulting magnetization within the particles is uniform. These assumptions have been experimentally verified in a number of papers \cite{Stanier2016, Lum2016, Hu2018}.

Finally, to model the size effects due to the presence of the magnetic fibers, we add a gradient energy density proportional to the $L^p$ norm of the second gradient of the deformation, with $p>2$. Summing up, we have the following regularized energy:
\begin{equation}\label{eq:3}
\mathcal E(y)=\int_{\Omega} W(x,\nabla y)+\mathfrak M(x,y,\nabla y)+\mu\ell^p|\nabla^{2} y|^p {\rm d}x_1{\rm d}x_2,
\end{equation}
where $\mu$ and $\ell$ are, respectively, a characteristic energy density, e.g., shear modulus, and a characteristic length. Materials whose stored-energy depends on the gradient of the deformation $F$ are examples of ``non-simple materials".
The higher-order  term incorporates scale effects in the mechanical properties which are usually relevant in micro-- and nano--structures (see for instance \cite{Evans2007,Duan2007}).
As we shall see below, the presence of the second gradients of the deformation bring about additional regularity, as well as compactness of the set of admissible deformations in a topology stronger than the weak $W^{1,p}$ topology, which is usually adopted in the approach to existence of solutions through the direct method.

\section{Existence of minimizers}

\subsection{The admissible set.}  Let $p>2$ be fixed. In the foregoing developments the following \emph{admissible sets} play
an important role:
\begin{equation}\label{eq:6}
  \mathcal{A}_c(\Omega)=\left\{{y} \in W^{2, p}\left(\Omega, \mathbb{R}^{2}\right)
    \textrm{ such that } \left.\begin{array}{l}
y(x)=x\text{ on }\Gamma,\\
 \displaystyle \|y\|_{W^{2,p}(\Omega;\mathbb R^2)}+\int_\Omega \frac 1 {|\det\nabla y|^q}{\rm d}x_1{\rm d}x_2\le c,
\\[0.8em]
\det\nabla {y}({x})>0\text { for a.e. } {x} \in {\Omega}, \\[0.5em]
\displaystyle \int_\Omega \det{\nabla y}\, {\rm d}x_1{\rm d}x_2\le\operatorname{meas}(y(\Omega))
\end{array}\right.\right\}.
\end{equation}
Here $c>0$ is a constant, $\Gamma\subset\partial\Omega$ is a part of the boundary having positive length, where the clamping condition $y(x)=x$ is enforced, and $q$ is an exponent satisfying the inequality
  \begin{equation}\label{eq:15}
    q\ge\frac{2p}{p-2}.
  \end{equation}
  Before proceeding further, we would like to comment on some properties of the admissible sets.

  \paragraph{Smoothness}  The first important property of the admissible set $\mathcal A_c(\Omega)$ is that each of its elements admits a continuous representative. Morrey's inequality \cite{Fucik-Kufner-John}, a fundamental result from the theory of function spaces, asserts that if a real-valued function $u$ belongs to the Sobolev space $W^{1,p}(U)$ with $U\subset\mathbb R^n$ a Lipschitz domain and with exponent $p$ greater than the dimension $n$ of the domain, then $u$ is H\"older continuous with exponent $\alpha=1-n/p$. In particular, $u$ is continuous up to the boundary of $U$. Thus, the equivalence class of $u$ in $W^{1,p}(U)$ has a continuous representative, whose values make sense pointwise in the closure of $U$. When this result is applied to the present case to the matrix-valued function $F=\nabla y$ where $y\in\mathcal A_c(\Omega)$, we have $F\in C^{0,1-2/p}(\overline\Omega;\mathbb R^{2\times 2})$. This implies the continuity of $\det F$ on $\overline\Omega$. For this reason, in the definition \eqref{eq:6}, the qualification of $\det F$ as a positive function holds in a pointwise sense.

  \paragraph{Local and global invertibility} Since $\nabla y$ is continuous, by the inverse-function theorem the condition that $\operatorname{det}\nabla y>0$  implies that $y$ is locally invertible at all points of $\Omega$. As is well known, local invertibility does not guarantee the global invertibility of the deformation $y$. As an example, consider a thin rod which is initially straight, and then is bent so that its ends superpose. To obtain global invertibility of the elements of $\mathcal{A}_c(\Omega)$, we have imposed the condition  $\int_\Omega\operatorname{det}\nabla y\le\operatorname{meas}(y(\Omega))$, which was introduced by Ciarlet and Ne\v{c}as in \cite{ciarletnecas} to ensure injectivity of the deformation. In fact, as one can easily see, this condition is violated whenever distant parts of the body undergo a superposition on a set of positive measure. This inequality improved on a previous result by Ball which stated that a deformation which is locally invertible is also globally invertible if it is invertible at the boundary. One of the  original reasons why the condition appeared first as an  inequality is that the set of deformations that satisfy the Ciarlet and Ne\v{c}as condition is weakly closed in $W^{1,p}(\Omega;\mathbb R^2)$, which is desirable if, when applying the direct method, one deals with weakly convergent sequences. In the present case, the reverse inequality 
  $\int_\Omega\operatorname{det}\nabla y\ge\operatorname{meas}(y(\Omega))$ holds \color{black}as a consequence of the change-of-variables formula\color{black}.

  \paragraph{Uniformly strictly positive determinant} 
A key property of the admissible sets is summarized in the following result.
  \begin{proposition}[Local invertibility, T.J. Healey and S. Kr\"omer \cite{healeykroemer}]\label{thr:1}
  Let $c>0$ be such that $\mathcal A_c(\Omega)\neq\emptyset$. Then there exists $\varepsilon=\varepsilon(c,p,q)>0$ such that
  \begin{equation}\label{eq:55}
  y\in\mathcal A_c(\Omega)\Rightarrow \textrm { det } \nabla y\geq \varepsilon \text { on } \overline{\Omega}.
\end{equation}
\end{proposition}
Besides \cite{healeykroemer}, a proof of Theorem \ref{thr:1} can be found in \cite[Theorem 2.5.3]{kruzikroubicek}. We point out that a key role in the proof is played by the fact that, according to the definition of $\mathcal A_c(\Omega)$ in \eqref{eq:6},
\begin{equation}\label{eq:333}
    y\in\mathcal A_c(\Omega)\quad\Rightarrow\quad \int_\Omega \frac{1}{|\det \nabla y|^q}{\rm d}x_1{\rm d}x_2<c.
  \end{equation}\color{black}
With this fact in mind, Theorem 1 follows through the application of the next lemma, which slightly generalizes  \cite[Thm.~3.1]{healeykroemer}.
\begin{lemma}\label{prop:1}
  Let $n\ge 1$ and $U$ be bounded Lipschitz domain in  $\mathbb R^n$. Let $f$ be a function in $C^{0,\alpha}(\overline U)$ such that $f>0$ a.e. in $U$  and 
  \begin{equation}
    \|f\|_{C^{0,\alpha}(\overline U)}\le {K}, \quad \int_U \frac 1 {|f|^q}{\rm d}x_1{\rm d}x_2\le {K},
  \end{equation}
  for some $K>0$ and $q\ge n/\alpha$. Then
  \begin{equation}
    f\ge \hat\varepsilon(U,\alpha,K\color{black},q)>0\quad \text{in}\quad\overline U,
  \end{equation}
  where the constant $\hat\varepsilon$ depends on the domain $U$ and on the constants $\alpha$, $K$, and $q$, and is independent on the particular function $f$.
\end{lemma}

\begin{proof}
Assume that the statement of the lemma is false and that there is 
$x\in\bar U$ such that $f(x)=0$.
Assume first that $x\in U$. 
We have that 
$$ |f(x)-f(y)|\le C|x-y|^\alpha$$
for every $y\in \bar U$. In other words, 
$|f(y)|^q\le C|x-y|^{\alpha q}$ because $f(x)=0$. 
Consequently, taking $B(x,r)\subset U$
\begin{align*}
    K\ge\int_U\frac{1}{|f(y)|^q}{\rm d}y\ge \int_{B(x,r)}\frac{1}{|f(y)|^q}{\rm d}y\\
    \ge \int_{B(x,r)}\frac{1}{C^qr^{\alpha q}}{\rm d}y
    =C{\rm meas}(B(0,1))r^{n-\alpha q}\ .
\end{align*}
If $q\alpha>n$ the last term tends to infinity for $r\to 0$ which gives the contradiction.
If $q\alpha=n$ then we have that 
$$
\int_{B(x,r)}\frac{1}{|f(y)|^q}{\rm d}y\ge C{\rm meas}(B(0,1))$$
for every $B(x,r)\subset U$ and this contradict uniform integrability of $\frac{1}{|f|^q}$. 
If $x\in\partial U$ then we proceed similarly taking into account that the Lipschitz property of $U$ implies that ${\rm meas}(B(x,r)\cap U)\ge Cr^n
$ for some $C>0$. We invite the interested reader to show that $\varepsilon$ does not depend on a particular $f$.
\end{proof}
To see how Lemma \ref{prop:1} applies to prove Theorem \ref{thr:1}, we use Morrey's inequality, applied to $\nabla y$, to obtain the implication
\begin{equation}\label{eq:56}
  y\in\mathcal A_c(\Omega)\quad\Rightarrow\quad\ \nabla y \in C^{0,1-2/p}(\overline\Omega;\R^{2\times 2})\ .
\end{equation}

Since $F\mapsto \det F$ is a separately convex function complying (in dimension 2) with the bound $|\det F|\le |F|^2$, we can apply \cite[Prop. 2.32]{dacorogna} to obtain
\begin{equation}
    |\det F(x_1)-\det F(x_2)|\le C(1+|F(x_1)|+|F(x_2)|)|F(x_1)-F(x_2)|,\qquad\text{for all }x_1,x_2\in\overline\Omega,
\end{equation}
 which implies H\"{o}lder continuity of $\det \nabla y$ and Lemma~\ref{prop:1} applies with $f=\det\nabla y$.
\begin{remark}\label{rem:22}
  When reading the statement of Lemma \ref{prop:1} one may wonder why the exponent $q$ should be greater than $n/\alpha$ and why the domain $\Omega$ is required to have the cone property. Concerning the exponent, suppose that $f$ vanishes at a point $x_0\in\overline U$, then the bound on its $C^{0,\alpha}$ norm would limit the growth $|f|$ in a sufficiently small neighborhood of $x_0$; as a result, the singularity of $1/{|f|^q}$ at $x_0$ would not be integrable for $q>n/\alpha$. This argument requires that the density of the set $U$ be positive at all points of the closure of $U$, which is guaranteed if $U$ has the cone property (i.e., when $\Omega$ is Lipschitz) but not, for instance, when $U$ has a cusp.
\end{remark}

\subsection{Assumptions}
We now state assumptions that: 1) guarantee that the functional $\mathcal E$ given by \eqref{eq:3} is well defined in the admissible space $\mathcal A_c(\Omega)$ for some positive $c$\color{black}; 2) render the functional coercive, \BBB which is one of the ingredients we will use to prove the existence of minimizers.\BBB 

First of all, we require the following:
\begin{align}
  a:\Omega\to\mathbb R^2 \textrm{ is measurable and }|a(x)|=1\textrm{ for a.e. }x\in\Omega.\label{eq:17}
\end{align}
Next, we ask the following conditions on $M$ and $h$:
\begin{subequations}\label{eq:2}
\begin{align}
  &M\in L^\infty(\Omega),\label{eq:7}\\
  &h:\mathbb R^2\to\mathbb R^2 \textrm{ is continuous and bounded}.\label{eq:18}
\end{align}
\end{subequations}
Finally, we impose the following coercivity condition\GGG{s}\BBB
  \begin{equation}\label{eq:13}
    \begin{aligned}
  &\exists c,C>0
      \text{ such that }W_{e}(F ; a \otimes a) \ge \displaystyle{c\left(|F|^{p}+\frac{1}{\left.|\det F\right|
      ^{q}}\right)-C}\quad\forall F\in\mathbb R^{2\times 2}:\
      \det F>0\!,
\\
& W_{e}(F ; a \otimes a) =+\infty \text{ if } \det F\le 0\!.
\end{aligned}
\end{equation}
These conditions guarantee that if $\mathcal E(y)$ is finite then $y\in\mathcal A_c(\Omega)$ for some $c>0$. Beside the coercivity condition, we also  impose a growth condition when $F$ has large norm and \emph{non-degenerate determinant}:
\begin{equation}\label{eq:54}
  \exists C_1>c\text{ such that $\forall \varepsilon>0$ }:W_{\rm e}(F;a\otimes a)\le C_1(|F|^p+1/\varepsilon^q)\quad\forall F\in\mathbb R^{2\times 2}:\det F \ge \varepsilon>0.
\end{equation}

We conclude this section with some remarks concerning the motivation for the above assumptions.

\begin{remark}
Assumption \eqref{eq:17} guarantees that  $W(\cdot,F)$ is a measurable function for every $F$ fixed, and that $W(x,\cdot)$ is continuous for a.e. $x\in\Omega$ \BBB (this is proved in the next section). \color{black}Accordingly, the integrand $W(\cdot,\cdot)$ in \eqref{eq:9} is a Caratheodory integrand. This is a typical requirement in variational problems. Assumptions \eqref{eq:7} and \eqref{eq:8} play a similar role, since they ensure that for $y\in\mathbb R^2$ and $F\in\mathbb R^2$ fixed the function $x\mapsto\mathfrak M(x,y,F)$ is measurable. In particular, we require that $h$ be continuous to guarantee that the composition $h\circ y$ with the measurable function $y$ is measurable. The requirement that $M$ and $h$ be bounded ensures that the integrand $\mathfrak M(x,y(x),\nabla y(x)$ besides being measurable, is bounded from below (see \eqref{eq:57} in the foregoing part of the paper). This fact, together with the fact that $W(x,F)$ is bounded from below, guarantees that the negative part of the argument of the integral in the definition of the energy is finite.
\end{remark}

\begin{remark}
  The growth condition \eqref{eq:54} shall be invoked in Step 3 of the proof of Theorem \ref{main-thm} below (see in particular the bounds \eqref{eq:45} and \eqref{eq:77}). It is similar to the condition (2.3) of \cite{bhatta1999}.
  \end{remark}

\begin{remark}
  The coercivity  assumptions \eqref{eq:13} on the energy $W$ along with the boundary condition on $\Gamma$ imply that sequences of deformations $y_k$ with uniformly bounded energy are uniformly bounded in $W^{1,p}(\Omega;\mathbb R^2)$. The extra regularizing term yields indeed a bound in $W^{2,p}(\Omega;\mathbb R^2)$. The fact that exponent $p$ is greater than 2, the number of dimensions, entails that the deformation gradients $F_k=\nabla y_k$ of the above sequence are also bounded in $C^{0,\lambda}(\overline\Omega;\mathbb R^{2\times 2})$ with $\lambda=1-2/p$. This fact is important because, when applying the direct method of the calculus of variations, we can pass to the limit in the integrand $W(x,F_k)$.
\end{remark}\color{black}

\begin{remark}\label{rem:2}
     The limit passage in the magnetic interaction energy $\mathfrak M$ (see \eqref{eq:53}) is more delicate, compared to that in the strain energy $W$. This is because the function $F\mapsto \mathfrak M(x,y,F)$ is not continuous and bounded, having  singularity at $F=0$ (note however that $F=0$ is not admissible). The property \eqref{eq:55} alleviates this problem, since it guarantees that $|Fa|$ is uniformly bounded by a positive constant for all $F=\nabla y$ with  $y\in\mathcal A_c(\Omega)$. This allows us to replace $F\mapsto \mathfrak M(x,y,F)$ with any extension which is continuous at $F=0$. To see this point, let $c$ be fixed, and consider any $y\in\mathcal A_c(\Omega)$. Define $F=\nabla y$ and let $F=RU$ be the polar decomposition of $F$. Since $U$ is symmetric, positive-definite it can be diagonalized. We denote by $\lambda_1$ and $\lambda_2$ the eigenvalues of $U$, and we sort them so that $\lambda_1\le \lambda_2$. Then $|Fa|=|RUa|=|Ua|\ge \lambda_1$. Moreover, $\lambda_1=\det F/\lambda_2\ge \varepsilon/\lambda_2$, where the last inequality follows from from Prop. \ref{thr:1}. \color{black}On the other hand $\lambda_2^2\le\lambda_1^2+\lambda_2^2=|U|^2=|F|^2\stackrel{\eqref{eq:56}}\le C$. Thus we conclude that $|Fa|\ge C$, as desired.
   \end{remark}
\begin{proposition}\label{prop:existence}
 Let \eqref{eq:17}-\eqref{eq:13} hold. Then $\mathcal E$ given in \eqref{eq:3} has a minimizer on $\mathcal{A}(\Omega)$.
\end{proposition}

\subsection{Well posedness}
Under the assumptions \eqref{eq:17}-\eqref{eq:13}, we show that if $y$ belongs to the admissible set $\mathcal A_c(\Omega)$ for some $c>0$, that is to say, $y\in\mathcal A(\Omega)$ where
\begin{equation}
  \mathcal A(\Omega)=\bigcup_{c>0}\mathcal A_c(\Omega),
\end{equation}
then the integral on the right-hand side in the definition \eqref{eq:3} of $\mathcal E(y)$ is well-defined.
\smallskip

\noindent We divide the proof in two steps. In the first step we check that the integrand in \eqref{eq:3} is measurable. In the second step we argue that the integral makes sense, according to the standard definition (see \emph{e.g.} \cite[\S4.7]{dibenedetto}).
\medskip

\emph{Step 1. Measurability.} Suppose that $y\in\mathcal A(\Omega)$. We are going to show that the function $x\mapsto W(x,F(x)\color{black})+\mathfrak M(x,y(x),F(x)\color{black})$ is measurable, where, as usual, we use the shorthand notation $F=\nabla y$.

As a start, recalling the definition \eqref{eq:9}, we notice that, since $W_{\rm e}$ is a continuous function, since $F$ is measurable by assumption, and since $a$ is measurable by $\eqref{eq:17}$, the composite function $x\mapsto W(x,F(x))=W_e(F(x),a(x)\otimes a(x))$ is measurable. It remains for us to check that the function
\begin{equation}\label{eq:4}
  x\mapsto \mathfrak M(x,y(x),F(x))=-M(x)\color{black}h(y(x))\cdot\frac{F(x)a(x)}{|F(x)a(x)|}
\end{equation}
is measurable. Indeed, since $h$ is a continuous function, it is Borel measurable; thus, the composition $h\circ y$ is Lebesgue measurable. By \eqref{eq:17} we have $(h\circ y)\cdot Fa$ is measurable; concerning the remaining term $|Fa|^{-1}$, we observe (\emph{cf.} Remark ) that $|Fa|>0$ a.e\color{black}. Then it follows that $|Fa|^{-1}$ is measurable (\cite[Chap. 4, Prop. 1.2]{dibenedetto}).

\smallskip

\emph{Step 2. Boundedness from below\color{black}.} Once we know that the integrand in \eqref{eq:3} is measurable, its integrability is guaranteed if either its positive or its negative part have finite integral according to the definition in \S4.7 of \cite{dibenedetto}. This is indeed the case. In fact, since $h$ and $M$ have been assumed to be bounded (\emph{cf.} \eqref{eq:2}), the function in \eqref{eq:4}
\BBB is bounded, that is,
\begin{equation}\label{eq:57}
  \BBB|\mathfrak M(x,y(x),\nabla y(x))|\le C\quad\text{for a.e. }x\in\Omega,
\end{equation}
for some constant $C>0$. Likewise, the function $x\mapsto W(x,F(x))$ is bounded from below.
Thus the integral of the negative part of the right-hand side of \eqref{eq:3} is finite. Accordingly, the entire integral is well defined (the integral may be in general infinite, unless we enforce a growth conditions on $W_e$).

\subsection{Existence of minimizers}\label{sec:existence-minimizers}
\BBB We split the proof of Proposition~\ref{prop:existence} in several steps.
\medskip

  \emph{Step 1. Infimizing sequence.}

  Let $y_0(x)=x$ for all $x\in\Omega$ denote the trivial deformation. Then $y_0\in\mathcal A(\Omega)$ and in particular $\mathcal E(y_0)<+\infty$. This implies that $m=\operatorname{inf}_{y\in \mathcal A(\Omega)}\mathcal E(y)<+\infty$. Thus, there exists sequence of deformations $y_k\in \mathcal A(\Omega)$ such that $\mathcal E(y_k)\to m$ as $k\to\infty$. \BBB Without loss of generality we can assume that the sequence $\mathcal E(y_k)$ of energies is monotone decreasing, with
  \begin{equation}\label{eq:58}
    \mathcal E(y_k)\le C
  \end{equation}
  for some constant $C$.

  \medskip

\emph{Step 2. Compactness.}
\BBB Using, in the order, the definition of $\mathcal E$ in \eqref{eq:3}, the definition of $W$ in terms of  $W_e$ given in \eqref{eq:9}, the coercivity of $W_e$ in \eqref{eq:13}, and the boundedness of $\mathfrak M$ in \eqref{eq:57}, we obtain\color{black}
\begin{equation}\label{eq:60}
  \|\nabla y_k\|^p_{L^p(\Omega;\mathbb R^{2\times 2})}\le C_1\mathcal E(y_k)-C_2,
\end{equation}
\BBB where $C_1$ and $C_2$ are positive constants. \color{black}Moreover, since
 $y_k(x)=x$ on a set of positive length (\emph{cf.} \eqref{eq:6}) \color{black}
it follows from the Poincar\'e's inequality that
 \begin{equation}\label{eq:8}
       \|y_k\|_{L^{p}\left(\Omega, \mathbb{R}^{2}\right)} \le C \|\nabla y_k\|_{W^{1,p}(\Omega;\mathbb R^{2\times 2})}.
      \end{equation}
Clearly, we have also
\begin{equation}\label{eq:59}
  \|\nabla\nabla y_k\|^p_{L^p(\Omega;\mathbb R^{2\times 2\times 2})}\le C_1\mathcal E(y_k)-C_2.
\end{equation}
By putting together \eqref{eq:58}--\eqref{eq:59} we conclude that the sequence $y_k$ is uniformly bounded in $W^{2,p}(\Omega;\mathbb R^2)$ and hence there exists a subsequence such that (following standard practice we do not relabel subsequences):\color{black}
\begin{equation}\label{eq:1004}
  y_k\rightharpoonup y\quad\text{weakly in }W^{2,p}(\Omega;\mathbb R^2).
\end{equation}\color{black}
Furthermore, by the Morrey embedding theorem $W^{2,p}(\Omega;\mathbb R^2)$ is compactly embedded in $C^{1,1-2/p}(\overline\Omega;\mathbb R^2)$. Thus, by further extracting a subsequence we have:
    \begin{equation}\label{eq:1}
      y_n\to y\qquad\textrm{strongly in }C^{1,1-2/p}(\overline\Omega;\mathbb R^2).
    \end{equation}
    \emph{Step 2: Admissibility of the limit.} We next verify that $y\in\mathcal A(\Omega)$. To begin with, we observe that
    \begin{equation}\label{eq:1005}
        y(x)=x\quad\text{for all }x\in\Gamma,
    \end{equation}
    since   $y_k(x)=x$ for all $x$ in $\Gamma$ and since $y_k$ converges pointwise to $y$ in $\overline\Omega$. \color{black}
    We henceforth use the notation
    \begin{equation}
      F=\color{black}\nabla y,\qquad F_k=\color{black}\nabla y_k.
    \end{equation}
By \eqref{eq:9}, \eqref{eq:13}, \eqref{eq:3}, and \eqref{eq:57} there exist positive constants $C_1$ and $C_2$ such that
    \begin{equation}
        \int_\Omega \frac {1} {|\det F_k|^q}\operatorname{d}x\le C_1\int_\Omega W(x,y_k(x))\le C_1 \mathcal E(y_k)-C_2.
    \end{equation}
    Thus, in view of \eqref{eq:3}, and by the definition of $\mathcal A_c(\Omega)$ in \eqref{eq:6}, we conclude $y_k\in\mathcal A_c(\Omega)$ for some $c$ fixed. As a result, we can apply Proposition \ref{thr:1} to obtain\color{black}
    \begin{equation}\label{eq:1001}
        \det F_k\ge \varepsilon,
    \end{equation}
    where $\varepsilon$ is positive constant. \color{black}By \eqref{eq:1} the sequence $F_k$ converges uniformly to $F$ in $\overline\Omega$. Thus, also $\det F_{k}\to\det F$ uniformly in $\overline\Omega$. In view of \eqref{eq:1001} this implies,
    \begin{equation}\label{eq:1003}
      \det F\ge \varepsilon\qquad\textrm{ in }\overline\Omega.
    \end{equation}
    The last inequality has the consequence that
\begin{equation}\label{eq:1002}
\int_{\Omega} \frac{1}{| \operatorname{det} F|^{q}}{\rm d}x_1{\rm d}x_2<+\infty.
\end{equation}
By the same token,
$$
\int_{\Omega} \det F{\,\rm d^2}x=\lim _{k \rightarrow\infty} \int_{\Omega} \det F_k\, {\,\rm d}^2x=\lim _{k \rightarrow \infty} \operatorname{meas}\left(y_{k}(\Omega)\right).
$$
On the other hand, a\color{black}n argument used in the proof of Theorem 5 of \cite{ciarletnecas} shows that
$$
\lim _{k \rightarrow \infty} \operatorname{meas}\left(y_{k}(\Omega)\right) \leq \operatorname{meas}(y(\Omega)).
$$
Thus, we conclude that
\begin{equation}\label{eq:25}
\int_{\Omega} \operatorname{det} F \,{\rm d}x_1{\rm d}x_2 \leqslant \operatorname{meas}(y(\Omega)).
\end{equation}
Recalling that $y\in W^{2,p}(\Omega;\mathbb R^2)$ by \eqref{eq:1004}, and by putting together \eqref{eq:1005}, \eqref{eq:1002}, and \eqref{eq:1003}, we obtain $y\in\mathcal A(\Omega)$, as desired. \color{black}

\emph{Step 3. Weak lower semicontinuity.} Having established that $y\in\mathcal A(\Omega)$, it remains for us to prove that
\begin{equation}\label{eq:1009}
  \mathcal E(y)\le\liminf_{k\to\infty}\mathcal E(y_k),
\end{equation}
which confirms that the infimum of the minimizing sequence is attained.

First observe that since $F_k\to F$ in $C^{0,1-2/p}(\overline\Omega;\mathbb R^2)$ with $\det F_k\ge \varepsilon>0$ (\emph{cf.} \eqref{eq:1} and \eqref{eq:1001}) we have that $W_e(F_k,a\otimes a)$ converges uniformly to $W_e(F,a\otimes a)$ in $\overline\Omega$, by the continuity of $W_e$. Accordingly, by bounded convergence,
\begin{equation}\label{eq:1008}
\begin{aligned}
    \int_\Omega W(x,F)\operatorname{d}^2x&=\int_\Omega W_e(F(x),a(x)\otimes a(x))\operatorname{d}^2x\\
    &=\lim_{k\to\infty}\int_\Omega W(F_k(x),a(x)\otimes a(x))\operatorname{d}^2x\\
    &=\lim_{k\to\infty}\int_\Omega W(x,F_k)\operatorname{d}^2x.
\end{aligned}
\end{equation}
Moreover, by the lower semicontinuity of the norm we have
\begin{equation}\label{eq:1007}
    \int_\Omega \mu\ell^p |\nabla\nabla y|^p{\rm d}x_1{\rm d}x_2\le\liminf_{k\to\infty}\int_\Omega \mu\ell^p |\nabla\nabla y|^p{\rm d}x_1{\rm d}x_2.
\end{equation}
Finally, we observe that the functions $F_ka$ convergence uniformly to $F a$ in $\overline\Omega$. Therefore, also $|F_ka|\to |Fa|$ uniformly.
 Moreover, as observed in Remark \ref{rem:2}, the bound \eqref{eq:1003} entails that $|F a|>0$.  As a result, we have the convergence
\begin{align}\label{eq:1006}
  \int_\Omega\mathfrak M(x,y,\nabla y)&=\int_\Omega (h\circ y)\cdot\frac{F a}{|F a|}{\rm d}x_1{\rm d}x_2=\lim_{k\to\infty} \int_\Omega (h\circ y_k)\cdot\frac{F_ka}{|F_ka|}{\rm d}x_1{\rm d}x_2\nonumber \\
  &=\lim_{k\to\infty}\int_\Omega\mathfrak M(x,y_k,\nabla y_k)\ .
\end{align}
By combining \eqref{eq:1008},\eqref{eq:1007}, and \eqref{eq:1006}, we obtain \eqref{eq:1009}.
\color{black}

\begin{remark}
 The regularizing term in the energy \eqref{eq:3} turns out to be a key ingredient in the proof of existence of a minimizer. In fact, in the absence of the regularizing energy the standard existence result for the equilibrium problem in nonlinear elastostatics \cite{ball1977} cannot be applied, since Ball's result requires the strain energy to be polyconvex. As illustrated by a counterexample in the Appendix, the magnetic density $\mathfrak M(x,y,\cdot)$ is not even rank-one convex, a property implied by polyconvexity (see for example \cite[Thm 5.3]{dacorogna}). 
  \end{remark}

\begin{remark}\label{rem:010}
  While the function $W$ is a Carath\'eodory integrand,  the function $\mathfrak M$ is not, since the mapping $F\mapsto Fa/|Fa|$ is singular at $F=0$\color{black}.  Thus, the standard lower semicontinuity arguments, such as those in \cite[Theorem 3.3.1]{kruzikroubicek} do not apply. On the other hand, thanks the coercivity condition \eqref{eq:13}, we can apply Proposition \ref{thr:1}, which ensures that  $\det F\ge\varepsilon>0$, i.e., the determinant is bounded from below by a positive (although small constant). In turns, this guarantees that the denominator $|Fa|$ stays above zero, and hence the singularity of $F\mapsto Fa/|Fa|$ at $F=0$ can be removed.\color{black}
  \end{remark}

  \begin{remark}
  The entire treatment of the existence result could have been done replacing the set $\mathcal A(\Omega)$ with the set
  \begin{equation}
      \mathcal A'(\Omega)=\{y\in W^{2,p}(\Omega;\mathbb R^3):y(x)=x\text{ on }\Gamma,\det\nabla y>0\text{ a.e. in }\Omega\}.
  \end{equation}
  In fact, if $y\in\mathcal A'(\Omega)$ then integral on the right-hand side of the definition  \eqref{eq:3} of the energy $\mathcal E(y)$ makes sense. If this integral is bounded by some constant $M$, then we have that $y\in\mathcal A_c(\Omega)$ for some $c>0$ that depends on $M$, which in turn implies that the determinant of $\nabla y$ is bounded from below by a positive constant that depends on $M$ but not on $y$.
\end{remark}

  \section{Rigorous dimension reduction}
  \subsection{Setup}
\paragraph{Geometry} We consider a family of thin bodies having the shape of a rectangle of length $\ell$ and thickness $t$. We can obtain these domains by considering a reference rectangle of unit height
\begin{equation}
  \widetilde\Omega=\left\{\left(x_{1}, x_{2}\right)\in\mathbb R^2: 0<x_{1}<\ell,-1/{2}<x_{2}<1/ 2\right\},
\end{equation}
and by introducing the projection map $\pi_t:\widetilde\Omega\to\Omega_t$ 
\begin{equation}
  \pi_t(x)=(x_1,tx_2),
\end{equation}
hence \BBB
\begin{equation}\label{eq:63}
\Omega_{t}=\pi_t(\widetilde\Omega)=(0,\ell)\times(-t/2,+t/2)\ .
\end{equation}
 \color{black}The symbol $t$, which stands for  ``thickness", is being used here in place of the more conventional symbol $h$, which we have used for the magnetic field.\color{black}

\paragraph{Constraints} We require that, in order to be admissible, a deformation $y_t$ be  equal to the identity on the set
$$
\Gamma_t=\pi_t(\widetilde\Gamma),\qquad \widetilde\Gamma=\{(0,x_2)\in\mathbb R^2:-1/2<x_2<+1/2\}\subset\partial\widetilde\Omega.
$$
that is, we impose
$$
y_t(x)=x\qquad\text{on }\Gamma_t.
$$
Here $y_t$ lives in the following class of \emph{admissible maps}
\begin{equation}\label{eq:6b}
  \widehat{\mathcal{A}}_c(\Omega_t)\equiv\left\{{y}_t \in W^{2, p}\left(\Omega_t, \mathbb{R}^{2}\right)
    \textrm{ such that } \left.\begin{array}{l}
y_t(x)=x\text{ on }\Gamma_t,\\
 \displaystyle \|y_t\|_{W^{2,p}(\Omega_t;\mathbb R^2)}+\int_{\Omega_t} \frac 1 {|\det\nabla y_t|^q}{\rm d}x_1{\rm d}x_2\le c,
\end{array}\right.\right\}
\end{equation}

\paragraph{Energy} For each $t$, we define the energy functional of the body $\Omega_t$ by
\begin{equation}\label{eq:12}
  \mathcal E_t(y_t)=\displaystyle\int_0^\ell \int_{-t/2}^{t/2} W(x_1,\nabla y_t)+\mu\ell^p|\nabla^{2} y_t|^p+\mathfrak{M}(x,y_t,\nabla y_t){\rm d}x_1{\rm d}x_2.
  \end{equation}
\begin{remark}{\rm
Here we are assuming the strain energy $W(\cdot,F)$ to be independent on the coordinate $x_2$. This assumption is not restrictive and may be relaxed, but its removal would provide little advantage at the cost of the introduction of additional technicalities.}
\end{remark}

\begin{figure}[H]
\begin{center}
\def\svgwidth{0.5\textwidth}
\begingroup%
  \makeatletter%
  \providecommand\color[2][]{%
    \errmessage{(Inkscape) Color is used for the text in Inkscape, but the package 'color.sty' is not loaded}%
    \renewcommand\color[2][]{}%
  }%
  \providecommand\transparent[1]{%
    \errmessage{(Inkscape) Transparency is used (non-zero) for the text in Inkscape, but the package 'transparent.sty' is not loaded}%
    \renewcommand\transparent[1]{}%
  }%
  \providecommand\rotatebox[2]{#2}%
  \newcommand*\fsize{\dimexpr\f@size pt\relax}%
  \newcommand*\lineheight[1]{\fontsize{\fsize}{#1\fsize}\selectfont}%
  \ifx\svgwidth\undefined%
    \setlength{\unitlength}{271.80969479bp}%
    \ifx\svgscale\undefined%
      \relax%
    \else%
      \setlength{\unitlength}{\unitlength * \real{\svgscale}}%
    \fi%
  \else%
    \setlength{\unitlength}{\svgwidth}%
  \fi%
  \global\let\svgwidth\undefined%
  \global\let\svgscale\undefined%
  \makeatother%
  \begin{picture}(1,0.7939167)%
    \lineheight{1}%
    \setlength\tabcolsep{0pt}%
    \put(0,0){\includegraphics[width=\unitlength,page=1]{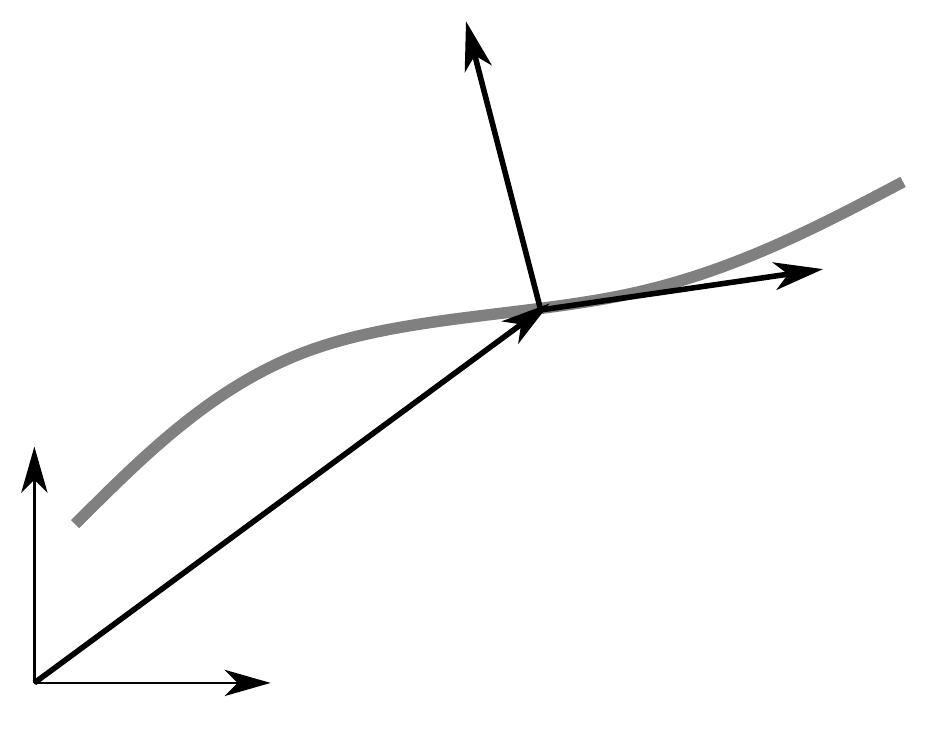}}%
    \put(0.29800895,0.07031448){\makebox(0,0)[lt]{\lineheight{1.25}\smash{\begin{tabular}[t]{l}$e_1$\end{tabular}}}}%
    \put(0.02294668,0.3404244){\makebox(0,0)[lt]{\lineheight{1.25}\smash{\begin{tabular}[t]{l}$e_2$\end{tabular}}}}%
    \put(-0.00131137,0.00774246){\makebox(0,0)[lt]{\lineheight{1.25}\smash{\begin{tabular}[t]{l}$\mathcal{O}$\end{tabular}}}}%
    \put(0.2263721,0.42200067){\makebox(0,0)[lt]{\lineheight{1.25}\smash{\begin{tabular}[t]{l}$z$\end{tabular}}}}%
    \put(0.58504914,0.40970544){\makebox(0,0)[lt]{\lineheight{1.25}\smash{\begin{tabular}[t]{l}$z(x_1)$\end{tabular}}}}%
    \put(0.85156311,0.44980135){\makebox(0,0)[lt]{\lineheight{1.25}\smash{\begin{tabular}[t]{l}$z'(x_1)$\end{tabular}}}}%
    \put(0.53274265,0.76515626){\makebox(0,0)[lt]{\lineheight{1.25}\smash{\begin{tabular}[t]{l}$b(x_1)$\end{tabular}}}}%
    \put(0,0){\includegraphics[width=\unitlength,page=2]{vector_fields.pdf}}%
    \put(0.73617061,0.6455986){\makebox(0,0)[lt]{\lineheight{1.25}\smash{\begin{tabular}[t]{l}$a(x_1)$\end{tabular}}}}%
  \end{picture}%
\endgroup%

\end{center}
\caption{Deformed configuration of the rod. The vector $z(x_1)$ contains the coordinates of the current position of point $x_1$, while $b(x_1)$ represents the cross section that contains $x_1$. The orientation of the magnetic fibers is given by $a(x_1)$.}
\label{vectorfields}
\end{figure}

\subsection{Statement of the the convergence result}\label{sec:convergence_result}
Having defined the equilibrium problem for a generic rectangular strip $\Omega_t$ of thickness $t$ (see \eqref{eq:63}), we choose for each $t$ a minimizer:
    \begin{equation}\label{eq:14}
      y_t\in\mathop{\operatorname{argmin}}_{\widehat{\mathcal A}(\Omega_t)}\mathcal E_t
    \end{equation}
    taken over the set $\widehat{\mathcal A}(\Omega_t)=\bigcup_{c>0}\widehat{\mathcal A}_c(\Omega_t)$.

    We would like to capture the asymptotic behavior of $y_t$ as $t\to 0$. To this effect, we need a way to compare deformations over different  domains. A convenient way to do so is to consider the thickness averages
 \begin{equation}\label{eq:33}
    z_t(x_1)=\frac 1 t \int_{-t/2}^{+t/2}y_t(x_1,x_2){\rm d}x_2,
  \end{equation}
  which are defined on the same domain $(0,\ell)$ \emph{independently of the thickness}. It turns out, however, that thickness averages alone do not allow to capture the asymptotic behavior of the deformation, and that in addition we need to record the following functions
 \begin{equation}\label{eq:33b}
    b_t(x_1)=\frac 1 t \int_{-t/2}^{+t/2}\partial_2 y_t(x_1,x_2){\rm d}x_2,
  \end{equation}
  which represent the average over the thickness of the image of a material fiber parallel to vector $e_2=(0,1)^T$. The convergence result can be stated rigorously in the next proposition. In that statement, we make reference to a generic sequence  $\{t_n\}_{n\in\mathbb N}$ such that $t_n\to 0$ as $n\to\infty$, and for typographical convenience we replace the symbol $t_n$ with $t$.
  \begin{proposition}\label{prop:conv-avg}
    Let $(\widetilde y_t)_{t>0}\subset\widehat{\mathcal A}(\Omega_t)$ be a sequence of minimizers of the energies $\mathcal E_t$ defined in \eqref{eq:12}. Then there is a (non-relabeled) sequence such that the averages defined in \eqref{eq:33} and \eqref{eq:33b} satisfy
  \begin{equation}\label{eq:34}
    \begin{aligned}
      &z_t\to z\quad\text{weakly in}\quad W^{2,p}((0,\ell);\mathbb R^2),\qquad \text{and}\\
      &b_t\to b\quad\text{weakly in}\quad W^{1,p}((0,\ell);\mathbb R^2),
     \end{aligned}
   \end{equation}
Moreover, the limit $(z,b)$ is a minimizer of the  functional $\mathcal F: W^{2,p}((0,\ell);\mathbb R^2)\times W^{1,p}((0,\ell);\mathbb R^2)\to\mathbb R\cup\{+\infty\}$  defined by 
  \begin{align}\label{eq:32}
    \mathcal{F}(z,b)=
      \displaystyle \int_0^{\ell} W(x_1,(z'\vert b))+\mathfrak M(x_1,z,(z'|b))  +\mu \ell^p(| z^{\prime\prime}|^2 +2|b'|^2)^{p/2}{\rm d}x_1
  \end{align}
   on the set 
  \begin{equation}\label{eq:21}
    \begin{aligned}
      \hypertarget{defB}{\hyperlink{defB}{\mathcal B}=\{&(z,b)\in W^{2,p}((0,\ell);\mathbb R^2)\times W^{1,p}((0,\ell);\mathbb R^2):\\
      &z(0)=0, b(0)=e_2, \det(z'|b)\ge\varepsilon, \text{for some }\varepsilon>0\}.}
\end{aligned}
  \end{equation}
\end{proposition}
Before proving the above result, some remarks are in order.
\begin{remark}
  The model emerging from dimension reduction features a non-simple one-dimensional continuum with an extra \emph{director field} $b$, which represents, physically, \emph{the rotation and stretch of the cross sections}.
\end{remark}
\begin{remark}
  Since the energy is non-convex, we cannot expect uniqueness of minimizers neither for the original 2-D energy nor for the derived 1-D energy. In particular, different choices of minimizers of the 2-D energy may be possible in principle, and these minimizers may possibly lead to different limits.
  \end{remark}
  \subsection{Proof of the convergence result.}
Proposition \ref{prop:conv-avg} is a consequence of a slightly stronger statement, which is proved in Theorem \ref{main-thm} below. The theorem in question is the main result of the present paper, and its statement requires the introduction of some further machinery from dimension reduction. 
    
As a start, we introduce the scaled deformation $\widetilde y_t=y_t\circ \pi_t$, that is,
\begin{equation}\label{eq:35}
\widetilde y_{t}(x_1,x_2)= y_t(x_1,tx_2)
\end{equation}
for all $x=(x_1,x_2)\in\widetilde\Omega$. Then
\begin{equation}\label{eq:50b}
  \mathcal E_t(y_t)=t\widetilde{\mathcal E}_t(\widetilde y_t),
\end{equation}
where $\widetilde{\mathcal E}_t:W^{2,p}(\widetilde\Omega;\mathbb R^2)\to\mathbb R$ is the functional defined by
\begin{equation}\label{eq:19}
  \widetilde{\mathcal E}_t(\widetilde y_t)=\int_0^\ell\int_{-1/2}^{+1/2}W(x_1,\nabla_t \widetilde y)+\mathfrak M(x_1,\widetilde y_t,\nabla_t\widetilde y_t)+\mu\ell^p|\nabla_t\nabla_t \widetilde y|^p {\rm d}x_1{\rm d}x_2.
  \end{equation}
where we have set
\begin{equation}\label{eq:20}
\nabla_{t} \widetilde y_t=\partial_1 {\widetilde y_t}\otimes e_1+\frac 1 t\partial_2 {\widetilde y_t}\otimes e_2=\left(\partial_{1} {\widetilde y_t} | t^{-1}  \partial_{2} {\widetilde y_t}\right),
\end{equation}
and
\begin{equation}\label{eq:22}
  \nabla_t\nabla_t {\widetilde y_t}=\partial_1 \nabla_t {\widetilde y_t}\otimes e_1+\frac 1 t\partial_2 \nabla_t {\widetilde y_t}\otimes e_2=\Bigg(\begin{array}{c|c} \partial^2_{11}{\widetilde y_t}&  t^{-1}\partial^2_{12}{\widetilde y_t}\\
                              \hline                                                                                      t^{-1}\partial^2_{21}{\widetilde y_t}& t^{-2}\partial^2_{22}{\widetilde y_t}\end{array}\Bigg).
\end{equation}

We can now state the main result. We recall that $\widehat{\mathcal A}(\Omega_t)=\bigcup_{c>0}\widehat{\mathcal A}_c(\Omega_t)$.
 \begin{theorem}\label{main-thm}
  Let $\{\widetilde{y}_{t}\}_{t>0}\subset \widehat{\mathcal{A}}(\widetilde\Omega)$ be a sequence of minimizers of $\widetilde{\mathcal E}_{t}$. Then there is a (non-relabeled)  subsequence such that 
  \begin{subequations}\label{eq:52}
  \begin{equation}\label{eq:52a}
     {\widetilde y_t} \rightharpoonup \widetilde y \quad \text { weakly in } W^{2, p}({\widetilde\Omega} ; \mathbb{R}^{2}), 
  \end{equation}
and
  \begin{equation}\label{eq:52b}
      \frac{\partial_2{\widetilde y_{t}}}{t}\rightharpoonup \widetilde  b\quad \text{ weakly in }W^{1,p}({\widetilde\Omega};\mathbb R^{2}).
  \end{equation}
  \end{subequations}
The functions $(\widetilde y,\widetilde b)$ do not depend $x_2$ and hence they can be identified with a pair $(z,b)$ of functions having  the interval $(0,\ell)$ as their common domain. This pair belongs to the set $\mathcal B$ defined in \eqref{eq:21} and minimizes, over the same set, the functional $\mathcal F$ defined in \eqref{eq:32}.
\end{theorem}
With Theorem \ref{main-thm} at hand, Proposition \ref{prop:conv-avg} follows by observing that the functions $(z,b)$ are indeed the over-the thickness averages of $(\widetilde y,\widetilde b)$ and that the convergence statements in \eqref{eq:34} are an immediate consequences of \eqref{eq:52a} and  \eqref{eq:52b}.

\subsection{Proof of the convergence result.}\label{sec:proof-conv-result-1}
We split the proof of Theorem \ref{main-thm} in three  steps.

\paragraph{Step 1. Compactness} Let ${\rm id}_t:\Omega_t\to\Omega_t$ be the identity map over $\Omega_t$. The trivial deformation ${\rm id}_t$ is admissible (i.e. it belongs to $\widehat{\mathcal A}_c(\Omega_t)$) for some $c=c(t)$ for every $t$. The corresponding rescaled deformation is $\widetilde{\rm id}_t={\rm id}\circ \pi_t$. We have $\widetilde{\mathcal E}_t(\widetilde{\rm id}_t)\stackrel{\eqref{eq:50b}}=t\mathcal E({\rm id}_t)=C_1\,t\,{\rm meas}(\Omega_t)\stackrel{\eqref{eq:63}}=C_2$, for some constants $C_1$ and $C_2$. Note that \eqref{eq:54} guarantees that $x_1\mapsto W(x_1,I)$ is bounded in $(-1/2,1/2)$. Therefore, the comparison between the energy of the rescaled minimizer $\widetilde y_t$ with that of the competitor $\widetilde{\rm id}_t$ yields
\begin{equation}\label{eq:26b}
  \widetilde{\mathcal E}_t(\widetilde y_t)\le C
\end{equation}
Since the Zeeman energy $\mathfrak M$ is bounded (recall definition \eqref{eq:53} and assumptions \eqref{eq:7} and \eqref{eq:18}), it follows from \eqref{eq:12} and \eqref{eq:26b} that
\begin{equation}\label{eq:26}
  \int_0^\ell\int_{-1/2}^{+1/2} W(x_1,\nabla_t\widetilde y_t)+\mu\ell^p|\nabla_t\nabla_t\widetilde y_t|^p{\rm d}x_1{\rm d}x_2\le C.
\end{equation}
      It follows from \eqref{eq:21}, from the definition \eqref{eq:9} of $W$, and from the coercivity assumption $\eqref{eq:13}_1$ on $W_e$ that
\begin{align}
\left\|\nabla_{t} {\widetilde y}_{t}\right\|_{L^{p}\left({\widetilde\Omega} ; \mathbb{R}^{2\times 2}\right)} \leq C,\label{eq:4.3}
\end{align}
for every $t>0$; it also follows that
\begin{equation}\label{eq:23}
  \|\nabla_t\nabla_t {\widetilde y}_t\|_{L^p({\widetilde\Omega};\mathbb R^{2\times 2\times 2})}\le C.
\end{equation}
On the other hand, by \eqref{eq:20},
\begin{equation}\label{eq:36}
\left\|\nabla {\widetilde y}_{t}\right\|_{W^{1,p}\left({\widetilde\Omega} ; \mathbb{R}^{2 \times 2}\right)} \leq\left\|\nabla_{t} {\widetilde y}_{t}\right\|_{W^{1,p}\left({\widetilde\Omega} ; \mathbb{R}^{2 \times 2}\right)};
\end{equation}
moreover, by \eqref{eq:22},
\begin{equation}\label{eq:24}
\left\|\nabla\nabla {\widetilde y}_{t}\right\|_{L^{p}\left({\widetilde\Omega} ; \mathbb{R}^{2 \times 2\times 2}\right)} \leq\left\|\nabla_t\nabla_{t} {\widetilde y}_{t}\right\|_{L^{p}\left({\widetilde\Omega} ; \mathbb{R}^{2 \times 2\times 2}\right)}.
\end{equation}
Since $\widetilde y_t$ is fixed on a part of the boundary having positive length (see \eqref{eq:6b}), the inequalities \eqref{eq:4.3}--\eqref{eq:24} imply that the sequence $(\widetilde y_t)_t$ is uniformly bounded in $W^{2,p}({\widetilde\Omega};\mathbb R^2)$. Thus, there exists $\widetilde y\in W^{2,p}({\widetilde\Omega};\mathbb R^2)$ and a subsequence of $(\widetilde y_t)_t$, which we do not relabel, such that
\begin{equation}\label{eq:pippo}
\begin{aligned}
  &{\widetilde y_t} \rightharpoonup \widetilde y \quad \text { weakly in } W^{2, p}({\widetilde\Omega} ; \mathbb{R}^{2}).
  \end{aligned}
\end{equation}
As a further consequence of \eqref{eq:4.3} and \eqref{eq:23}, we have the bounds
\begin{equation}
  \left\|\frac{\partial_2\widetilde y_t}{t}\right\|_{L^p({\widetilde\Omega};\mathbb R^2)}\le C,\qquad \left\|\frac{\partial_{22}\widetilde y_t}{t^2}\right\|_{L^p({\widetilde\Omega};\mathbb R^2)}\le C.\label{bds}
\end{equation}
The first of these bounds along with \eqref{eq:pippo} and connectedness of the vertical cross-sections of  $\widetilde\Omega$    imply that
\begin{equation}
  \begin{aligned}
    &\partial_2 \widetilde y=0 \quad\text{ a.e. in }{\widetilde\Omega},
    \end{aligned}
  \end{equation}
  and that there exists $\widetilde b\in W^{1, p}({\widetilde\Omega} ; \mathbb{R}^{2})$ such that
  \begin{equation}\label{eq:5}
      \frac{\partial_2{\widetilde y_t}}t\rightharpoonup \widetilde  b\quad \text{ weakly in }W^{1,p}({\widetilde\Omega};\mathbb R^{2});
  \end{equation}
the second bound  in \eqref{bds} implies that the weak limit $\widetilde b$ is independent on $x_2$:
\begin{equation}\label{eq:pluto2}
  \begin{aligned}
  &\partial_2 \widetilde  b=0\quad\textrm{ a.e. in }{\widetilde\Omega},
  \end{aligned}
\end{equation}
and that there exists $\widetilde  d\in W^{1,p}({\widetilde\Omega};\mathbb R^2)$ such that
\begin{equation}\label{eq:pluto}
  \begin{aligned}
  &\frac{\partial_{22}^{2} {\widetilde y_t}}{t^{2}} \rightharpoonup  d \quad \text { weakly in } L^{p}({\widetilde\Omega} ; \mathbb{R}^{2}). 
  \end{aligned}
\end{equation}
Summing up,
\begin{equation}\label{eq:10}
  \begin{aligned}
     \nabla_t \widetilde y_t&=\partial_1\widetilde y_t\otimes e_1+\frac{\partial_2\widetilde y_t}t\otimes e_2=\left({\partial_1\widetilde y_t}\mid\frac{\partial_2\widetilde y_t}t\right)\rightharpoonup (\partial_1 \widetilde y|\widetilde b)\quad\text{ weakly in }W^{1,p}({\widetilde\Omega};\mathbb R^{2\times 2}),\qquad\text{and}\\
    \nabla_t^2 \widetilde y_t&=\partial_{11}\widetilde y_t\otimes e_1\otimes e_1+\frac{\partial_{12} \widetilde y_t}t\otimes (e_1\otimes e_2+e_2\otimes e_1)+ \displaystyle {\frac{\partial_{22}^{2} \widetilde y_{t}}{t^{2}}}\otimes e_2\otimes e_2\\
    &=\left(   \begin{array}{c|c}
               \partial_{11}\widetilde y_t & \displaystyle \frac{\partial_{12}\widetilde y_t}t\\
                             \\
               \hline
                             \\
         \displaystyle\frac{\partial_{21}\widetilde y_t}t &  \displaystyle {\frac{\partial_{22}^{2} \widetilde y_{t}}{t^{2}}}
             \end{array}\right)
           \rightharpoonup
           \left(   \begin{array}{c|c}
                      \partial_{11}\widetilde y & \displaystyle {\partial_{1} \widetilde b}\\
                             \\
\hline
                             \\
          \displaystyle {\partial_{1} \widetilde b} &  \displaystyle  d
    \end{array}\right)
  \quad\text{ weakly in } L^p({\widetilde\Omega};\mathbb R^{2\times 2\times 2}).
\end{aligned}
\end{equation}
It remains for us to show that $(\widetilde y,,\widetilde b)$ minimizes $\mathcal F$. This is done in the next two steps.

\paragraph{Step 2: Liminf inequality} Thanks to the first of \eqref{eq:10}, since $p>2$, Morrey's embedding theorem yields
\begin{equation}\label{eq:27}
  \nabla_t\widetilde y_t\to (\partial_1 \widetilde y|\widetilde b)\quad \text{ strongly  in } C^{0,1-2/p}(\overline{\widetilde\Omega};\mathbb R^{2\times 2}).
\end{equation}
By the same token, since $h$ is a continuous function (\emph{cf.} \eqref{eq:18}), we have
\begin{equation}
    h\circ \widetilde y_t\to  h\circ \widetilde y \quad\text{ pointwise  a.e. in } \widetilde\Omega.
  \end{equation}
   Note that $W$ is bounded from below  and   $W(\cdot,\nabla_t\widetilde y_t)$ converges pointwise to $ W(\cdot, \partial_1\widetilde y|b)$ \BBB in $\widetilde\Omega$. Hence, by the Fatou lemma and by \eqref{eq:26},
  \begin{equation}\label{eq:37}
    \int_{\widetilde\Omega} W\big(x_1,(\partial_1\widetilde y|\widetilde b)\big){\rm d}x_1{\rm d}x_2\le \liminf_{t\to 0}  \int_{\widetilde\Omega} W(x_1,\nabla_t\widetilde y_t){\rm d}x_1{\rm d}x_2\le C.
  \end{equation}
  Likewise, by the weak lower semicontinuity of the $L^p$ norm, if follows from \eqref{eq:26} that
  \begin{equation}\label{eq:11}
  \begin{aligned}
\int_{\widetilde\Omega} \Big(|\partial_{11}y|^2+2|\partial_1 \widetilde b|^2\Big)^{\frac p 2}{\rm d}x_1{\rm d}x_2\le \int_{\widetilde\Omega} \left|  \left(   \begin{array}{c|c}
\partial_{11}\widetilde y & \displaystyle \partial_{1} \widetilde b\\
\\
\hline
\\
{\rm sym} &  \displaystyle d
\end{array}\right)\right|^p{\rm d}x_1{\rm d}x_2\stackrel{\eqref{eq:10}}\le \liminf_{t\to 0}\int_{\widetilde\Omega} |\nabla_t^2 \widetilde y_t|^p{{\rm d}x_1{\rm d}x_2}\ .
\end{aligned}
\end{equation}
As $\widetilde y$ and $\widetilde b$ are both independent of $x_2$, we identify them with $z\in W^{2,p}(0,\ell;\R^2)$ and $b\in W^{1,p}(0,\ell;\R^2)$, respectively. This yields 
\begin{equation}\label{eq:29}
  \int_{{\widetilde\Omega}}W\big(x_1,(\partial_{1} \widetilde y \mid \widetilde b)\big) {\rm d} x_{1} {\rm d} x_{2}=\int_0^\ell W\big(x_1,(z'|b)\big){\rm d}x_1,
\end{equation}
and
\begin{equation}\label{eq:30}
  \begin{aligned}
  \int_{{\widetilde\Omega}}\big(|\partial_{11} \widetilde y|^{2}+2|\partial_{1} \widetilde{b}|^{2}\big)^{\frac p 2} {\rm d} x_{1} {\rm d} x_{2}=\int_{0}^{\ell}(\left|z^{\prime \prime}\right|^{2}+2\left|b^{\prime}\right|^{2})^{\frac p 2} {\rm d} x_{1}.
  \end{aligned}
\end{equation}
In view of definition \eqref{eq:9} and hypothesis $\eqref{eq:13}_1$, the bound \eqref{eq:37} entails that $\det(z'|b)>0$ a.e. in $(0,\ell)$ with
\begin{equation}
  \int_0^\ell \frac{1}{(\det(z'|b))^q}{\rm d}x_1<+\infty.
\end{equation}
Since $(z'|b)\in W^{1,p}((0,\ell);\mathbb R^{2\times 2})$,  we have $(z'|b)\in  C^{0,\alpha}([0,\ell];\mathbb R^{2\times 2})$ for $\alpha=(p-1)/p$. As a consequence, $\det(z'|b)\in C^{0,\alpha}([0,\ell])$. Since by  our  hypothesis $q>2p/(p-2)$ (\emph{cf.} \eqref{eq:15}), we have also $q>2/\alpha$. Thus Lemma \ref{prop:1} can be applied with $n=1$, $U=(0,\ell)$ and $f=\det(z'|b)$ to obtain
\begin{equation}\label{eq:62}
  \det(z'|b)\ge \varepsilon
\end{equation}
for some $\varepsilon>0$; then, by \eqref{eq:27}, we have
\begin{equation}\label{eq:61}
  \det(\nabla_t\widetilde y_t)\ge\frac \varepsilon 2,
\end{equation}
provided that $t$ is sufficiently small. This implies, in turn, that $\mathfrak M(x_1,\widetilde y_t,\nabla_t\widetilde y_t)$ converges uniformly to $\mathfrak M(x_1,\widetilde y,(\partial_1\widetilde y|\widetilde b)$ in $\widetilde\Omega$ (a fact that can be  established through an argument similar to that leading to \eqref{eq:25}, see also Remark \ref{rem:2}), and hence
\begin{equation}\label{eq:40}
  \int_{\widetilde\Omega} {\mathfrak M}(x_1,\widetilde y_t,\nabla_t\widetilde y_t){\rm d}x_1{\rm d}x_2\to \int_{\widetilde\Omega} {\mathfrak M}(x_1,\widetilde y,(\partial_1\widetilde y|\widetilde b){\rm d}x_1{\rm d}x_2=\int_0^\ell \mathfrak M(x_1,z,(z'|b)){\rm d}x_1
\end{equation}
as $t\to 0$. By putting  \eqref{eq:37}--\eqref{eq:30} and \eqref{eq:40} together, we have the ``liminf inequality'':
\begin{equation}\label{eq:42}
  \mathcal F(z,b)\le\liminf\widetilde{\mathcal E}_t(\widetilde y_t).
\end{equation}
To conclude the proof that $(z,b)$ is a minimizer of $\mathcal F$, we shall show in the next step that for every competitor $(\widehat z,\widehat b)\in\hyperlink{defB}{\mathcal B}$ there exists a recovery sequence $(\widehat y_t)_t$ such that
\begin{equation}\label{eq:41}
  \limsup_{t\to 0}\widetilde{\mathcal E}_t(\widetilde y_t)\le \mathcal F(\widehat z,\widehat b).
\end{equation}
Then, the combination of \eqref{eq:42}, the fact that $\widetilde y_t$ minimizes $\widetilde{\mathcal E}_t$, and \eqref{eq:41} yields the desired result:
\begin{equation}\label{eq:50}
  \mathcal F(z,b)\stackrel{}\le\liminf_{t\to 0}\widetilde{\mathcal E}_t(\widetilde y_t)\le\limsup_{t\to 0}\widetilde{\mathcal E}_t(\widehat y_t)\le\mathcal F(\widehat z,\widehat b).
\end{equation}

\paragraph{Step 3 Recovery sequence}
Since $(\widehat z,\widehat b)\in\hyperlink{defB}{\mathcal B}$, there is a constant $\varepsilon>0$ such that
    \begin{equation}\label{eq:43}
      \det (\widehat z'|\widehat b)\ge \varepsilon.
    \end{equation}
We take a sequence $(\widehat b_t)_t\subset C^\infty((0,\ell);\mathbb R^2)$ such that
\begin{equation}\label{eq:28}
\widehat b_t(0)=e_2,\qquad \widehat b_{t} \rightarrow \widehat b\quad \text{strongly in}\quad W^{1,p}((0,\ell);\mathbb R^2).
\end{equation}
Thanks to the smoothness of $\widehat b_t$, the functions
$$
\widehat y_t\left(x_1, x_2\right)=\widehat z\left(x_1\right)+t x_{2} \widehat b_t\left(x_1\right)
$$
belong to the space $W^{2,p}({\widetilde\Omega};\mathbb R^2)$. Moreover, if we identify $\widehat z$ and $\widehat b$ with functions defined on $\widetilde\Omega$, then, thanks the the third of \eqref{eq:28},
\begin{equation}\label{eq:31}
  \nabla_t\widehat y_t=(\widehat z'+tx_2\widehat b'_t|\widehat b_t)\to (\widehat z'|\widehat b)\quad\text{strongly in }W^{1,p}({\widetilde\Omega};\mathbb R^{2\times 2}),
\end{equation}
and
\begin{equation}\label{eq:46}
  \nabla_t\nabla_t\widehat y_t= \left(   \begin{array}{c|c}
\widehat z''+tx_2 \widehat b_t'' & \displaystyle \widehat b_t'\\
\\
\hline
\\
{\rm sym} &  0
\end{array}\right)\to \left(   \begin{array}{c|c}
\widehat z'' & \displaystyle \widehat b'\\
\\
\hline
\\
{\rm sym} &  0
\end{array}\right)\quad\text{strongly in }L^p({\widetilde\Omega};\mathbb R^{2\times 2\times 2}).
\end{equation}
By \eqref{eq:31} and by Vitali's convergence theorem, the functions $|\nabla_t\widehat y_t|^p$ are uniformly equi-integrable. Thus, for every $\varepsilon>0$ there exists $\delta(\varepsilon)>0$ such that
\begin{equation}\label{eq:39}
  \operatorname{meas}N\le\delta(\varepsilon)\quad\Rightarrow\quad \int_B\left(|\nabla_t\widehat y_t|^p+|(\widehat z'|\widehat b)|^p\right){\rm d}x_1{\rm d}x_2<\varepsilon
\end{equation}
for every measurable set $N\subset\widetilde\Omega$. Moreover, there is a subsequence such that $\nabla_t\widehat y_t$ converges almost everywhere. Since $W(x_1,\cdot)$ is continuous, we have also that the functions $W(x_1,\nabla_t\widehat y_t)$ converge almost everywhere and hence in measure. As a consequence of this fact, for $\varepsilon>0$ fixed, the measure of the set
\begin{equation}\label{eq:44}
N(\varepsilon,t)=\{x\in\widetilde\Omega:|W(x_1,\nabla_t\widehat y_t)-W(x_1,(\widehat z'|\widehat b)|>\varepsilon\},
\end{equation}
satisfies
\begin{equation}\label{eq:51}
   \operatorname{meas}N(\varepsilon,t) \to 0 \text{ as } t\to 0.
\end{equation}
By \eqref{eq:39}, \eqref{eq:51}, and by the growth assumption $\eqref{eq:54}$, we have
\begin{equation}\label{eq:45}
  \limsup_{t\to 0}\int_{N(\varepsilon,t)}\left|W(x_1,\nabla_t\widehat y_t)|+|W(x_1,(\widehat z'|\widehat b))\right|{\rm d}x_1{\rm d}x_2\le C\varepsilon
\end{equation}
for a suitably large constant $C$. Accordingly, we have
\begin{align}\label{eq:77}
    \limsup_{t\to 0} &\left|\int_{\widetilde\Omega}\left(W(x_1,\nabla_t\widehat y_t)-W(x_1,(\widehat z'|\widehat b)\right){\rm d}x_1{\rm d}x_2\right|\nonumber\\
    &\le
    \limsup\int_{\widetilde\Omega}\big|W(x_1,\nabla_t\widehat y_t)-W(x_1,(\widehat z'|\widehat b)\big|\,{\rm d}x_1{\rm d}x_2 \nonumber \\
&\le\limsup\int_{\widetilde\Omega\setminus N(\varepsilon,t)}\big|W(x_1,\nabla_t\widehat y_t)-W(x_1,(\widehat z'|\widehat b))\big|\,{\rm d}x_1{\rm d}x_2\nonumber\\
&\quad +\limsup\int_{N(\varepsilon,t)}\big|W(x_1,\nabla_t\widehat y_t)\big|+\big|W(x_1,(\widehat z'|\widehat b))\big|\,{\rm d}x_1{\rm d}x_2
\nonumber\\
&\le \varepsilon\operatorname{meas}(\widetilde\Omega)+C\varepsilon.
\end{align}
Since the constant $\varepsilon$ can be taken arbitrarily small, we conclude that
\begin{equation}\label{eq:47}
  \int_{\widetilde\Omega} W(x_1,\nabla_t\widehat y_t){\rm d}x_1{\rm d}x_2\to\int_{\widetilde\Omega} W(x_1,(\widehat y'|\widehat b)){\rm d}x_1{\rm d}x_2=\int_0^\ell W(x_1,(\widehat y'|\widehat b)){\rm d}x_1.
\end{equation}
Thus, taking also in to account \eqref{eq:46}, we have:
\begin{equation}\label{eq:49}
  \int_{\widetilde\Omega} W(x_1,\nabla_t\widehat y_t)+\mu\ell^p|\nabla_t\nabla_t\widehat y_t|^p{\rm d}x_1{\rm d}x_2\to \int_{\widetilde\Omega} W(x_1,(\widehat z'|\widehat b))+\mu\ell^p(|\widehat z''|^2+2|\widehat z'|^2)^{\frac p 2}{\rm d}x_1{\rm d}x_2.
\end{equation}
A further consequence of \eqref{eq:31} is that $\nabla_t\widehat y_t$ converges to $(\widehat z'|\widehat b)$ uniformly in $\overline{\widetilde\Omega}$; therefore, by \eqref{eq:43}, $\det(\nabla_t\widehat y_t)\ge\varepsilon/2$ for $t$ sufficiently small. This means that for each $(x,y)$ fixed, the function $\mathfrak M(x,y,\cdot)$ (which is continuous on the set of matrices satisfying $\det F\ge \varepsilon/2$, see Remark \ref{rem:2}\BBB), can be replaced by any continuous extension, and the argument used to arrive at \eqref{eq:47} can be replicated to conclude that
\begin{equation}\label{eq:48}
  \int_{\widetilde\Omega} {\mathfrak M}(x_1,\widehat y_t,\nabla_t\widehat y_t){\rm d}x_1{\rm d}x_2\to\int_0^\ell \mathfrak M(x_1,\widehat z,(\widehat z'|\widehat b)){\rm d}x_2.
\end{equation}
The combination of \eqref{eq:49} and \eqref{eq:48}, in view of the definition \eqref{eq:19}, yields the desired result \eqref{eq:41}, and by \eqref{eq:50}, the proof is concluded.

\subsection{Some technical remarks about the proof.}\label{sec:concluding}

In carrying out the dimensional reduction, the limit passages \eqref{eq:1006}, \eqref{eq:40}, and \eqref{eq:48}) where we handle the magnetic energy $\mathfrak M(x,y,\nabla y)$ which appears in  definitions \eqref{eq:3}, \eqref{eq:12}, and \eqref{eq:32}), are key in our proofs. The magnetic energy, as defined in \eqref{eq:53}, is indeed discontinuous. However, as already pointed out in Remark \ref{rem:2}, such discontinuity can be removed once we know that the Jacobian (the determinant of the deformation gradient) is bounded from below by a positive constant. The coercivity  \eqref{eq:13} of the strain energy with respect to the Jacobian is key to obtain this result through Lemma \ref{prop:1} and Proposition \ref{thr:1}.

In particular, Proposition \ref{thr:1} is instrumental in the proof of the  existence of a minimizer given in Section \ref{sec:existence-minimizers}, where it is used to  obtain the bound \eqref{eq:1001}, which in turn allows us to obtain the convergence \eqref{eq:1006}. Lemma \ref{prop:1} is instead invoked  in the proof of the liminf inequality (Step 3) in Section \ref{sec:proof-conv-result-1} to establish the bound \eqref{eq:62} on the limit gradient $(z|b)$ (thanks to our choice of the admissible set $\mathcal B$ in \eqref{eq:21}), and then, by continuity, the bound  $\eqref{eq:61}$ on the scaled gradients $\nabla_t\widetilde y_t$, which eventually leads to the limit passage  \eqref{eq:40}. Note carefully that \eqref{eq:61} cannot be obtained from Proposition \ref{thr:1}, since \emph{it involves the determinant of the rescaled gradients} $\nabla_t\widetilde y_t$, whereas Proposition \ref{thr:1} is based on a bound for the standard gradient $\nabla \widetilde y_t$ (note that we do not have the inequality $|\det{\nabla_t\widetilde y_t}|\le|\det{\nabla\widetilde y_t}|$!). A similar argumentation is used to exploit \eqref{eq:50} to arrive at \eqref{eq:48}.

Finally, we observe that the proofs of existence and convergence rely only on the convexity and the $p$-growth, with $p>2$, of the regularizing term $\mu\ell^p|\nabla^2 y|^p$ in \eqref{eq:3}. In particular, this term may be replaced by $\mu\Psi(\nabla^2 y)$ with a convex function $\Psi$ with $p$-th growth.

\section{Equilibrium equations}
In this section we derive the strong form of the equations governing the mechanical equilibrium of the rod. 
With reference to the terms under integral sign in the definition \eqref{eq:32} of the functional $\mathcal F$, we observe that by frame indifference there exists an isotropic function $\widehat W_{\rm e}(G;a\otimes a)$ such that
\begin{equation}\label{eq:69}
    W_e(F;a\otimes a)=\widehat W_{\rm e}(G;a\otimes a),\qquad G=\frac 12 (F^TF-I).
\end{equation}
In particular, for $F=z'\otimes e_1+b\otimes e_2$, the Green-Lagrange strain tensor is $G=\frac 12 (|z'|^2-1) e_1\otimes e_1+\frac 12 z'\cdot b(e_1\otimes e_2+e_2\otimes e_1)+\frac 1 2 (|b|^2-1)e_2\otimes e_2$. Therefore, on letting
\begin{equation}
    \widehat W(x_1,G)=\widehat W_{\rm e}(G;a\otimes a),
\end{equation}
and
\begin{equation}\label{eq:97b}
w(x_1,\eta,\zeta,\theta\color{black})=\widehat W\Big(x_1,\frac 1 2(\eta-1) e_1\otimes e_1+\frac 12 \zeta(e_1\otimes e_2+e_2\otimes e_1)+\frac 1 2 (\theta-1)e_2\otimes e_2\color{black}\Big),
\end{equation}
we have
\begin{equation}
    W(x_1,(z'|b))=w(x_1,|z'|^2,b\cdot z',|b|^2).
\end{equation}
Thus, upon setting
\begin{equation}\label{eq:99c}
    \mathfrak m(x_1,h,z',b)=Mh\cdot\frac{(z'|b)a}{|(z'|b)a|}\quad (=\mathfrak M(x_1,h,(z'|b))),
\end{equation}
we can write,
\begin{equation}\label{eq:102f}
    \mathcal F(z,b)=\int_0^L w(x_1,|z'|^2,b\cdot z',|b|^2\color{black})-\mathfrak m(x_1,h\circ z,z',b)+\mu\ell^p(|z''|^2+2 |b'|^2)^{\frac p 2}{\rm d}x.
\end{equation}
The requirement that $\mathcal F$ be stationary for a virtual variation $(\widetilde z,\widetilde b)$ of its arguments can be given the form of a \emph{virtual-work equation}:
\begin{equation}\label{eq:99b}
    \int_0^L n\cdot\widetilde z'+r\cdot\widetilde z''+q\cdot\widetilde b +m\cdot\widetilde b'\,{\rm d}x_1=\int_0^L f_{\rm m}\cdot\widetilde z+n_{\rm m}\cdot\widetilde z'+q_{\rm m}\cdot\widetilde b\,{\rm d}x_1.
\end{equation}
The left-hand side of \eqref{eq:99b} is the virtual work performed by the internal forces over the variation $(\tilde z,\tilde b)$. In particular, the force-like quantities $n,r,q,m$ are given by
\begin{equation}\label{eq:98a}
\begin{aligned}
    & n=2 w_{,2}(x_1,|z'|^2,b'\cdot z',|b|^2)z'+w_{,3}(x_1,|z'|^2,b'\cdot z',|b|^2)b,\\
    &q=w_{,3}(x_1,|z'|^2,b'\cdot z',|b|^2)z'+2w_{,4}(x_1,|z'|^2,b'\cdot z',|b|^2)b\\
            & r=p\mu\ell^p(|z''|^2+2|b'|^2)^{\frac p2-1}z'',\\
    & m=2p\mu\ell^p(|z''|^2+2|b'|^2)^{\frac p 2-1}b',
    \end{aligned}
\end{equation}
with $w_{,i}$ denoting the partial derivative of $w$ with respect to its $i$-th argument. The right-hand side of \eqref{eq:99b} is the work of the magnetic field, with $f_{\rm m}$, $n_{\rm m}$, and $q_{\rm m}$ being given by
\begin{equation}\label{eq:104bb}
    \begin{aligned}
    &f_{\rm m}=\nabla h^T(z)\mathfrak m_{,2}(x_1,h\circ z,z',b),\\
    &n_{\rm m}=\mathfrak m_{,3}(x_1,h\circ z,z',b),\\
    &q_{\rm m}=\mathfrak m_{,4}(x_1,h\circ z,z',b).
    \end{aligned}
\end{equation}
It is immediate from \eqref{eq:98a} that $n$ and $q$ are not independent, but obey
\begin{equation}\label{eq:102b}
    n\times z'+q\times b=0,
\end{equation}
Here the symbol $\times$ is an antisymmetric scalar-valued product defined as follows: for $a$ and $b$ a pair of vectors, $a\times b=Ra\cdot b$, with $R=\begin{pmatrix}0 &-1\\ 1 &0\end{pmatrix}$ being the counter-clockwise rotation by $\pi/2$. This property is a manifestation of frame indifference through the requirement that the virtual work be not affected by a rigid virtual variations to $(\tilde z,\tilde b)$.

By imposing that \eqref{eq:99b} holds for every admissible variation $(\tilde z,\tilde b)$  we arrive at the equilibrium equations:
\begin{equation}\label{eq:107gg}
\begin{aligned}
    &-(n-r')'=f_{\rm m}-n_{\rm m}',\\
    &-m'+q=q_{\rm m},\color{black}.
\end{aligned}
\end{equation}
By first taking the cross product with $b$, and then by taking the scalar product with $b$, we split the second of \eqref{eq:107gg} in two scalar components, orthogonal and parallel to $b$. Making use of \eqref{eq:102b}, we replace $b\times q$ with $-n\times q$, to arrive at
\begin{equation}\label{eq:107d}
\begin{aligned}
    &-(n-r')'=f_{\rm m}-n_{\rm m}',\\
    &-b\times m'+z'\times n=b\times q_{\rm m},\\
    &\phantom{-}b\cdot(-m'+q)=b\cdot q_{\rm m}.\color{black}
\end{aligned}
\end{equation}

\section{Linearization}\label{sec:linearization}
With a view towards gaining some insight into the mechanical implications of our theory, we seek a specialization to a linear version involving small departures from the \emph{reference configuration} $\mathring z(x_1)=x_1e_1$, $\mathring b(x_1)=e_2$, where $e_1$ and $e_2$ is the canonical basis of $\mathbb R^2$. In order for the regularization to be effective in this regime we shall henceforth take the regularizing exponent $p$ equal to $2$. Although this limit case is not encompassed by the existence theory developed in the previous sections, the latter theory could be generalized in the spirit of the discussion at the end of Subsection \ref{sec:concluding}.

Since $(\mathring z'|\mathring b)=I$, and since the identity tensor $I$ is a minimum of $W_{\rm e}(\cdot;a\otimes a)$,  the reference configuration is a minimum of the elastic part of the energy
\begin{equation}
    \mathcal F_{e}(z',b)=\int_0^L w(x_1,|z'|^2,b'\cdot z')+\mu\ell^2(|z''|^2+2|b'|^2)\,{\rm d}x_1.
\end{equation}
This being said, if $\mathring h$ is a spatially constant applied field and $a(x_1)$ is parallel to $\mathring h$ for all $x_1\in(0,L)$, then the reference configuration is a stationary point of the magnetic part of the energy:
\begin{equation}
    \mathring{\mathcal F}_{m}(z',b)=\int_0^L \mathfrak m(x_1,\mathring h\circ z,z',b')\,{\rm d}x_1.
\end{equation}
As a result, the reference configuration is an equilibrium configuration and the superposition of a small perturbation $\widetilde h$ to $\mathring h$ results into the applied field
\begin{equation}
    h=\mathring h+\tilde h,
\end{equation}
and into a small change of configuration that can be described through
\begin{equation}\label{eq:111c}
\begin{aligned}
    &z(x_1)=x+w(x_1)e_1+v(x_1)e_2,\\
    &b(x_1)=e_2-\phi(x_1)e_1+\psi(x_1)e_2,\color{black}
    \end{aligned}
\end{equation}
with $w$ and $v$ small compared to the length $L$ of the rod, and $w'$, $v'$, $\phi$, and $\psi$ small dimensionless quantities. The quantities $u$ and $v$ represent, respectively, the axial and the transversal displacement, whereas $\phi$ and $\psi$ represent, respectively, the \emph{counter-clockwise} rotation and the stretch of the typical cross section.

The fields $w$, $v$, and $\phi$ are obtained from the solution of the system resulting from the linearization of the equilibrium equations \eqref{eq:107d} and the constitutive equations \eqref{eq:98a}--\eqref{eq:104bb}. The linearized constitutive equations, in particular, depend on a set of constants resulting from the quadratic expansion of the strain energy in the undeformed configuration. To identify these constant, we make recourse to the representation formula
\begin{equation}\label{eq:112b}
\widehat W_{\rm e}(G;a\otimes a)=\widetilde W_{\rm e}(J_1(G),J_2(G),J_4(G,a\otimes a)),
\end{equation}
for the isotropic function $W_{\rm e}$, where
\begin{equation}\label{eq:113d}
\begin{aligned}
&J_1(G)=\operatorname{Tr}(G),&& J_2(G)=\det(G)=\frac 12 \big[J_1^2(G)-\operatorname{Tr}(G^2) \big], && J_4(G,a\otimes a) = \operatorname{Tr}(Ga\otimes a)
\end{aligned}
\end{equation}
are the appropriate invariants for the type of symmetry into play. Here we borrow the notation from three dimensional elasticity, where for transversely isotropic materials the constitutive response can be expressed in terms of a function depending on five invariants $J_i$, $i=1\ldots 5$. However, in 2D the number of independent invariants reduces to three since $J_2=J_3$ and $J_5=\text{Tr}(G^2a\otimes a)$ can be expressed in terms of $J_1$, $J_2$ and $J_4$ through to the Cayley-Hamilton theorem.

From the linearization of \eqref{eq:102b}
one obtains
\begin{equation}\label{eq:67}
   q_1=n_2
\end{equation}
As a result, the virtual-work equation \eqref{eq:99b} takes the form
\begin{equation}\label{eq:115fb}
  \begin{aligned}
    &\int_0^L n_1 \widetilde w'+n_2(\widetilde v'-\widetilde \phi)+p_1\widetilde w''+p_2\widetilde v''-m_1\widetilde\phi'+q_2\widetilde\psi+m_2\widetilde\psi'\color{black}{\rm d}x_1\\
  &\qquad\qquad=\int_0^L f_{\rm m1}\widetilde w+f_{\rm m2}\widetilde v+ n_{\rm m1}\widetilde w'+n_{\rm m2}\widetilde v'-q_{\rm m1}\widetilde\phi+q_{m2}\widetilde\psi \color{black}\,{\rm d}x_1.
  \end{aligned}
\end{equation}
Here, the application of \eqref{eq:98a}, \eqref{eq:111c}, \eqref{eq:112b}, and \eqref{eq:113d} yields the linear constitutive equations:
\begin{equation}\label{eq:116e}
\begin{aligned}
    &n_1=N_{11} w'+N_{12}(v'-\phi)+N_{13}\psi,\\
    &n_2=N_{21}w'+N_{22} (v'-\phi)+N_{23}\psi,\\
    &q_2=Q_{21}w'+Q_{22} (v'-\phi)+Q_{23}\psi,
    \end{aligned}
\end{equation}
where
\begin{equation}\label{eq:16}
\begin{aligned}
  &N_{11}=2\mu+\lambda+a_1^4\alpha_1+2a_1^2\alpha_2,\\
  &N_{13}=\lambda+a_2^2 \alpha_2 + a_1^2 (a_2^2 \alpha_1 + \alpha_2),\\
  &N_{23}=a_1 a_2 (a_2^2 \alpha_1 + \alpha_2 ),\\
  &N_{21}=N_{12}=a_1a_2(a_1^2\alpha_1+\alpha_2),\\
  &N_{22}=\mu+a_1^2a_2^2\alpha_1,\\
    &Q_{21}=\lambda+a_2^2 \alpha_2 + a_1^2 (a_2^2 \alpha_1 + \alpha_2),\\
    &Q_{22}=a_1a_2 (a_2^2 \alpha_1 + \alpha_2),\\
    &Q_{23}=2\mu+\lambda+a_2^4 \alpha_1 + 2 a_2^2 \alpha_2,
\end{aligned}
\end{equation}
with
\begin{equation}\label{eq:118e}
\begin{aligned}
&\mu=-\widetilde W_{\rm e,2}(0,0,0)/2,&&\lambda=\widetilde W_{\rm e,1,1}(0,0,0)+\widetilde W_{\rm e,2}(0,0,0),\\
    &\alpha_1=\widetilde W_{\rm e,4,4}(0,0,0), &&
    \alpha_2=\widetilde W_{\rm e,1,4}(0,0,0),   
\end{aligned}
\end{equation}
where the subscript $i$ denotes partial differentiation with respect to the $i$-th argument. Furthermore, we have
\begin{equation}\label{eq:66}
\begin{aligned}
        &p_1=2\mu\ell^2 w'',\qquad p_2=2\mu\ell^2 v'',\\
        &m_1=-4\mu\ell^2\phi',\qquad m_2=4\mu\ell^2\psi'.\color{black}
\end{aligned}
\end{equation}
Finally, we have
\begin{equation}\label{eq:64}
\begin{aligned}
&f_{\rm m}=M\nabla \widetilde h^T a,\\
&n_{\rm m}=a_1M(I-a\otimes a)\widetilde h-a_1^2M(\mathring h\cdot a)(I-a\otimes a)(we_1+ve_2)
\\&\qquad\quad-a_1a_2M(\mathring h\cdot a)(I-a\otimes a)(-\phi e_1+\psi e_2),\\
&q_{\rm m}=a_2M(I-a\otimes a)\widetilde h-a_2a_1M(\mathring h\cdot a)(I-a\otimes a)(we_1+ve_2)\\
&\qquad\quad -a_2^2M(\mathring h\cdot a)(I-a\otimes a)(-\phi e_1+\psi e_2).
\end{aligned}
\end{equation}
We split each of the equilibrium equations \eqref{eq:107gg} into a component parallel to $e_1$ and a component parallel to $e_2$. By doing so, and by making use of \eqref{eq:67}, we obtain:
\begin{equation}\label{eq:115f}
\begin{aligned}
  &-(n_1-p_1')'=f_{\rm m1}-n_{\rm m1}',\\
  &-(n_2-p_2')'=f_{\rm m2}-n_{\rm m2}',\\
    &-m'_1-n_2=q_{\rm m1},\\
    &-m'_2+q_2=q_{\rm m2}.\color{black}
    \end{aligned}
  \end{equation}
The above splitting is convenient in the special cases when the magnetic fibers are either parallel or orthogonal to the axis. In both cases the product $a_1a_2$ vanishes and the coupling coefficients $N_{21}$, $N_{12}$, $N_{23}$, $Q_{22}$ vanish as well. As a result, the system \eqref{eq:115f} admits a further splitting into two pairs of equations: the first and the fourth, which rule unknowns $w$ and $\psi$ (axial displacement and transverse stretch); the second and the third, which govern $v$ and $\phi$ (transverse displacement and rotation). We shall make use of this observation in the next section, where we examine an equilibrium problem for a rod whose magnetic fibers are parallel to the axis.

\section{Size-dependent magneto-elastic buckling}\label{sec:buckl-induc-magn}
We consider a cantilever beam whose reference configuration is shown in the following figure.
\begin{figure}[H]
\begin{center}
\def\svgwidth{0.7\textwidth}
   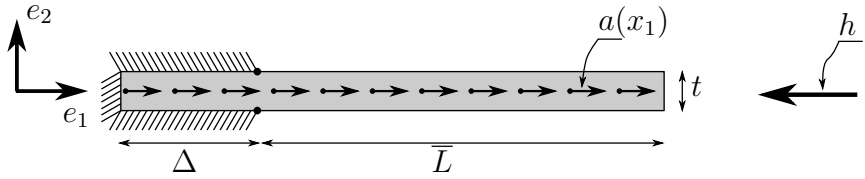
\end{center}
\caption{A cantilever with magnetic fibers parallel to the axis is acted upon by a uniform magnetic field whose direction is opposite to that of fibers. Although the model we deal with is one-dimensional, it is represented as two-dimensional body to better show its geometrical features.}
\label{geomlinear}
\end{figure}
In the reference configuration the cantilever has length $L=\overline L+\Delta$ and spans the interval $(-\Delta,L)$. The cantilever is clamped by a constraining device that, in the limit $t\to 0$, imposes null displacement and null rotation in the interval $(-\Delta,0)$. This arrangement guarantees not only that displacement and rotation vanish at the point $x_1=0$, but also that the displacement gradient vanishes at the same point, since displacement must be continuous together with its derivative.
As also shown in the above figure, the direction $a$ of the fibers is assumed to be parallel to the axis $e_1$, and the magnetic field $\mathring h$ has opposite direction with respect to $a$:
\begin{equation}\label{eq:65}
    a=e_1,\qquad \mathring h=-He_1,\qquad \widetilde h=0.
  \end{equation}
  We also assume the magnetization $M$ to be constant.

  As pointed out in Section \ref{sec:linearization}, the reference configuration is a stationary point of the energy $\mathcal F$ defined in \eqref{eq:102f}, and small departures from this configuration, as described by \eqref{eq:111c}, satisfy the virtual work equation \eqref{eq:115f}. With the loading environment considered in the present case, the magnetic energy has in fact a maximum, as opposite to the strain energy, which has a local minimum. As one can expect intuitively (an intuition that is confirmed in the foregoing) the applied field plays the same role as a compressive thrust, so that there exists a critical value of the applied field which renders the undeformed configuration unstable.

To make the argument quantitative, we first deduce the system of equations governing the unknowns $v,w,\phi$, and $\psi$ introduced in \eqref{eq:111c}. To begin with, we work out the  linearized formulae \eqref{eq:64} for the magnetic loads $f_{\rm m}$, $n_{\rm m}$, and $q_{\rm m}$; using \eqref{eq:65} we obtain:
\begin{equation}
    f_{\rm m}=0,\qquad n_{\rm m}=HM v'e_2,\qquad q_{\rm m}=0.
  \end{equation}
This, the magnetic loading appears only in the second equation of $\eqref{eq:115f}$; as anticipated at the end of the previous section, this equation contains only the unknowns $v$ and $\phi$; indeed, by making use of the linearized constitutive equations \eqref{eq:116e}--\eqref{eq:66}, the second of $\eqref{eq:115f}$ becomes:
  \begin{subequations}\label{eq:125f}
  \begin{equation}
      2\mu\ell^2 v^{(4)}-\mu (v'-\phi)'+HM v''=0.
    \end{equation}
Using \eqref{eq:116e}--\eqref{eq:66} again, the third of \eqref{eq:115f} becomes:
\begin{equation}\label{eq:68}
  4\mu\ell^2\phi''+\mu(v'-\phi)=0.
\end{equation}
\end{subequations}
The substitution of \eqref{eq:116e}--\eqref{eq:66} in the remaining equilibrium equations (namely, the first and the last of \eqref{eq:115f}) yield differential equations that involve the axial displacement $w$ and the transverse stretch $\psi$, and to not contain neither  $v$ nor $\phi$. These equations have only the trivial solution, and hence can be ignored. We can therefore focus on \eqref{eq:125f}.  

We seek a solution of these equations in the interval $(\Delta,L)$, supplemented by the essential conditions:
\begin{equation}\label{eq:126e}
\begin{aligned}
    v(\Delta)=0,\qquad \phi(\Delta)=0,\qquad v'(\Delta)=0,
\end{aligned}
\end{equation}
and by the natural conditions:
\begin{equation}\label{eq:70}
    \begin{aligned}
    &n_2(L)-n_{2m}(L)-p_2'(L)=0,\\
    &p_2(L)=0,\\
    &m_2(L)=0.
    \end{aligned}
    \qquad\Leftrightarrow\qquad
    \begin{aligned}
    & \mu(v'(L)-\phi(L))-HMv'(L)-2\mu\ell^2v'''(L)=0,\\
    & 2\mu\ell^2v''(L)=0,\\
    &-4\mu\ell^2\phi'(L)=0.
    \end{aligned}
\end{equation}
As already point out, the term proportional to $HM$ has a destabilizing effect, as a quadratic expansion of the energy would confirm. The mechanical interpretation is the following: the magnetic fibers would like to be aligned with the applied field, but their rotation is hindered by the stiffness of the structure. Buckling occurs when the destabilizing effect of the magnetic torques equals the stabilizing effect of the rod stiffness.

The critical value of the intensity of the magnetic field are determined by imposing that the linearized equilibrium problem has non-trivial solutions. To this aim, we integrate \eqref{eq:125f} to obtain:
\begin{equation}\label{eq:72}
  2\mu\ell^2 v'''-\mu(v'-\phi)+HM v'=c_1.
\end{equation}
From the first of \eqref{eq:70} we deduce $c_1=0$. Using \eqref{eq:68} in \eqref{eq:72} and integrating once we obtain:
\begin{equation}\label{eq:74}
4 \mu \ell^{2} \phi' +H M v+2 \mu l^{2} v'' =c_{2}.
\end{equation}
Using \eqref{eq:125f} to eliminate the unknown $\phi$ from \eqref{eq:74} and integrating once to arrive at the following equation:
\begin{equation}\label{eq:71}
  {2\ell^2}\left((3\mu -2HM)v''-4 \mu\ell^2  v''''\right)+HM v=c_2,
\end{equation}
where $c_2$ is a constant. The general solution of \eqref{eq:71} is
\begin{equation}\label{eq:73}
  v(x)=C_0+\sum_{i=1}^4C_i\exp(\lambda_i x_1/\ell),
\end{equation}
where $\lambda_i$, $i=1\ldots 4$, are the roots of the characteristic polynomial of \eqref{eq:71}, given by
\begin{equation}
  \lambda_i=\pm\frac{1}{2\sqrt{2}}\sqrt{\left(3 -2 \overline{HM}\right)\pm\sqrt{\left(3 -2 \overline{HM}\right)^2+2\overline{HM}}}
\end{equation}
with $\overline{HM}=HM/\mu$. The possible values of the constants $C_i$, $i=0\ldots 4$ in \eqref{eq:73} are then filtered by the three boundary conditions in \eqref{eq:126e}, and the second and the third boundary conditions in \eqref{eq:70}, which result into a linear system whose coefficients depend on the triplet $(HM/\mu,\ell/L)$. The system is singular when a characteristic equations of the form
\begin{equation}
  \Phi\left(\frac{HM}{\mu},\frac{\ell}{L}\right)=0
\end{equation}
holds true, with $\Phi$ a function whose expression has been determined using {\texttt{Mathematica}\textregistered} \cite{Mathematica}. We used the same package to performs a numerical root finding procedure. Table \ref{tab:1} shows the renormalized critical value $(HM)_c$ with $\ell/L$ ranging from $0.1$ to $0.2$ with step-size $0.01$, which confirms size-dependent behavior, with increasing values of material scale $\ell$ associated to increasing critical values of the magnetic field.
\begin{table}[H]\label{tab:1}
\centering
\begin{tabular}{|
  | c |c c c c c c c c c c c||} 
 \hline
 $\ell/L$ & 0.10& 0.11& 0.12& 0.13& 0.14& 0.15& 0.16& 0.17& 0.18& 0.19 & 0.20 \\ [0.5ex] 
 \hline
  ${(HM)}_{c}/\mu$ & 0.46& 0.47& 0.49& 0.51& 0.53& 0.56& 0.58& 0.61& 0.63& 0.66& 0.68\\
 \hline
\end{tabular}
\caption{Size dependence of the critical value $HM_c$.}
\label{table:1}
\end{table}

\color{black}\section{Numerical verification}
The weak form equation \eqref{eq:99b} of the proposed model was implemented into the \texttt{COMSOL Multiphysics \textregistered} software \cite{comsol}, the objective of this being the numerical verification of the convergence of the 2D-model for $t\rightarrow 0$ to the 1D one.\color{black}

The geometry of the specimen  is the same as in Fig.~\ref{geomlinear}. However, the magnetic field has different orientation, as shown in Fig.~\ref{geom} \color{black}with the 2D domain defined by
\[
\Omega_t = (0,L)\times\Big(-\frac t 2, \frac t 2\Big).
\]
For the reader's convenience we show the arrangement again in the following figure:
\begin{figure}[H]
\begin{center}
\def\svgwidth{1.0\textwidth}
   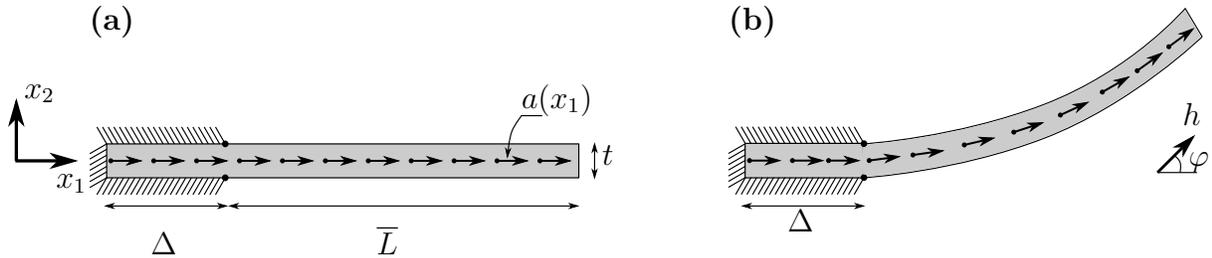
\end{center}
\caption{Geometry parameters of the actuators in its reference (a) and deformed (b) shapes. $h$ is the external magnetic field.}
\label{geom}
\end{figure}

A condition $u=0$ is specified in the region $\partial\Omega_u=\partial \Omega_t\cap\lbrace(x_1,x_2)\in\Omega_t | x_1\leq \Delta\rbrace$ so that $\Delta$ is the width of the clamped region and $\overline{L}=L-\Delta$ is the actual length of the rod. As already pointed out in Sec.~\ref{sec:buckl-induc-magn}, such a clamping condition is used to impose in the 2D model a condition equivalent to
\[
z(\Delta)=0,\qquad z'(\Delta)=0,\qquad b(\Delta)=e_2
\]
which is indeed used in the 1D model. 

The magnetic part of the energy is the one in Eq.~\eqref{eq:53} with 
\[
M(x) = M,\qquad h = H \cos(\varphi)e_1+ H \sin(\varphi) e_2
\]
i.e., the magnetization is assumed constant and $H$ is the intensity of the magnetic field.

For the elastic part of the energy, we focus on the widely used Saint Venant--Kirchhoff model (SV-K) with a transversely isotropic elasticity tensor; this choice is legitimate, since the presence of the higher-order term in \eqref{eq:3} renders the problem well posed even if the strain-energy density $W_{\rm e}(F; a \otimes a)$, is not polyconvex. In addition the SV-K energy depends on the same number of constitutive parameters of the linear theory: the four material parameters $\mu$, $\lambda$, $\alpha_1,$, $\alpha_2$, which we have defined in terms of $\widehat W_{\rm e}$ through \eqref{eq:69}, \eqref{eq:112b}, and \eqref{eq:118e}. The SV-K energy can be equivalently written as (see also \cite{Weiyi1999}):\color{black}
\begin{equation}
  \begin{aligned}
  W_{\rm e}(F; a\otimes a)=\widehat W_{\rm e}(G;a\otimes a),\qquad G=\frac 12(F^TF-I),
  \end{aligned}
\end{equation}
where
\begin{equation}
  \widehat W_{\rm e}(G;a\otimes a)=\Big(\frac \lambda 2 +\mu\Big) J_1^2(G)\color{black}-2\mu\, J_2(G)\color{black}+\frac{\alpha_1}{2} J_4^2(G,a\otimes a)\color{black}+\alpha_2 J_1(G)\color{black} J_4(G,a\otimes a)\color{black}\color{black},\\
\end{equation}
The first Piola-Kirchhoff stress tensor for this model can be obtained as $P=F\partial W_e/\partial G$, i.e.,
\begin{align}
P\color{black}=&\lambda J_1\,  F+2 \mu\, F G+ \alpha_1 J_4 FA + \alpha_2\big(J_1 F A+ J_4 F\big),\color{black} \label{SKModel}
\end{align}
where $A=a \otimes a$.

In order to implement the regularizing strain gradient term in the 2D model, we follow \cite{Horak2020} and we introduce an auxiliary second-order tensor field $\chi$ and its gradient. Subsequently, the continuum constraint $\chi=F$ is enforced weakly using penalty approach\footnote{The constraint $\chi=F$ was implemented in \texttt{Comsol} through the option \emph{Weak Constraint}, rather than directly defining it in the energy. The lagrangean multiplier $H_\chi$ was automatically determined by the software.}. Accordingly, for the 2D formulation we consider the following \emph{enhanced} functional
$$
\check {W}\color{black}({F}, {\chi}, {\nabla} \chi; a \otimes a)=W_{\rm e}(F; a \otimes a)+\frac{1}{2} H_{\chi}(F-\chi)^{2}+\frac{1}{2} K \,\nabla \chi\cdot {\nabla} \chi
$$
in which $H_{\chi}$ represents a new material parameter acting as a penalization forcing the new auxiliary field $\chi$ to remain as close as possible to $F$, whence for the limit case $H_{\chi} \rightarrow \infty$ a standard strain-gradient model is recovered.  When dealing with such a strain gradient regularization,  it is well known that using the same interpolation of the displacement and of the auxiliary field $\chi$ may lead to locking, in the sense that the results may become insensitive to the internal length parameter and strongly dependent upon the penalty parameter. This is mainly due to the incompatibility between finite element approximation and kinematic requirements that links $\chi$ to $F$ which has the same order as the gradient of displacement. To overcome this issue, Hermite quartic polynomials were used for the interpolation of the (first)gradient term, whereas  Argyris quintic polynomials for the second gradient part. In addition, the number of elements $N$ used for the discretization were set to be $
N = \displaystyle \operatorname{Min}\lbrace  \ell, t\rbrace$/3,
\color{black}in order to guarantee an accurate approximation of the solution when either the internal length $\ell$ or the thickness $t$ becomes very small. For all simulations, the geometry of the specimen was set to $\Delta=0.3 \overline L$, and the following values of the constitutive parameters were used
\[
\lambda = 10\mu\color{black},\quad \alpha_1=\alpha_2= \frac\mu 5
\]
for the (first)gradient part, 
\[
K_\chi =2\, \mu\,\ell^2, \quad \ell=0.1 {\overline L}\,.
\]
for the second gradient terms, whereas
\[
H M = \mu
\]
is the equivalent magnetic stiffness.

We assess the numerical convergence of the model in three cases:
\begin{enumerate}[(a)]
    \item \label{a} a purely mechanical problem (H=0) with an external body force applied in the $e_1$ direction,
    \item \label{b} for inclusions uniformly distributed, such that $a(x_1)=e_1$ and $\varphi=\frac \pi 3$,
    \item \label{c} for inhomogeneously oriented inclusions, $a(x_1)=cos(\theta)e_1+\sin(\theta)e_2$, $\theta=\frac{\pi}{2}+\frac{\pi}{\overline{L}}\,(x_1-\Delta)$ and $\varphi=\pi$.
\end{enumerate}

The results of the simulations for case \ref{a} are shown in Fig.~\ref{Fig0} . In this simplified analysis, no reinforcing fibres were considered, i.e., $a_1=a_2=0$, and a body force with intensity $q=\mu$ was applied in the direction $e_1$. As seen from the figure, the results are rather insensitive to the thickness of the model and in fact the different curves obtained for $t/\overline{L}\in\lbrace 0.001,0.01,0.05,0.1\rbrace$ are almost indistinguishable. The 1D model perfectly matches the results of the numerical 2D simulations. Figure~\ref{Fig0}.1 shows that the director field $b_2$ obtained from the numerical simulations.

\begin{figure}[H]
\begin{footnotesize}
\begin{center}
\def\svgwidth{0.3\textwidth}
\begingroup%
  \makeatletter%
  \providecommand\color[2][]{%
    \errmessage{(Inkscape) Color is used for the text in Inkscape, but the package 'color.sty' is not loaded}%
    \renewcommand\color[2][]{}%
  }%
  \providecommand\transparent[1]{%
    \errmessage{(Inkscape) Transparency is used (non-zero) for the text in Inkscape, but the package 'transparent.sty' is not loaded}%
    \renewcommand\transparent[1]{}%
  }%
  \providecommand\rotatebox[2]{#2}%
  \newcommand*\fsize{\dimexpr\f@size pt\relax}%
  \newcommand*\lineheight[1]{\fontsize{\fsize}{#1\fsize}\selectfont}%
  \ifx\svgwidth\undefined%
    \setlength{\unitlength}{366.46761676bp}%
    \ifx\svgscale\undefined%
      \relax%
    \else%
      \setlength{\unitlength}{\unitlength * \real{\svgscale}}%
    \fi%
  \else%
    \setlength{\unitlength}{\svgwidth}%
  \fi%
  \global\let\svgwidth\undefined%
  \global\let\svgscale\undefined%
  \makeatother%
  \begin{picture}(1,0.12810632)%
    \lineheight{1}%
    \setlength\tabcolsep{0pt}%
    \put(0,0){\includegraphics[width=\unitlength,page=1]{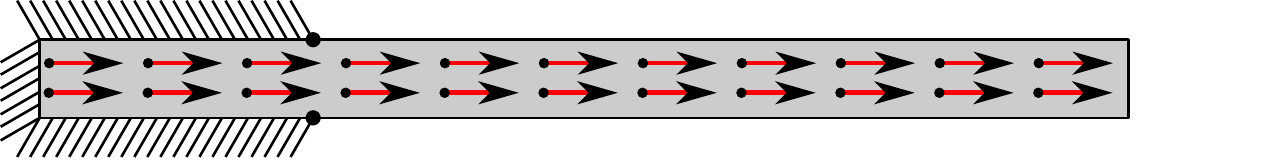}}%
    \put(0.9557511,0.00627185){\makebox(0,0)[lt]{\lineheight{1.25}\smash{\begin{tabular}[t]{l}$q$\end{tabular}}}}%
    \put(0,0){\includegraphics[width=\unitlength,page=2]{simplified_results.pdf}}%
    \put(0.9557511,0.00627185){\makebox(0,0)[lt]{\lineheight{1.25}\smash{\begin{tabular}[t]{l}$q$\end{tabular}}}}%
    \put(0,0){\includegraphics[width=\unitlength,page=3]{simplified_results.pdf}}%
  \end{picture}%
\endgroup%
\\[0.2cm]
\includegraphics[width=8cm]{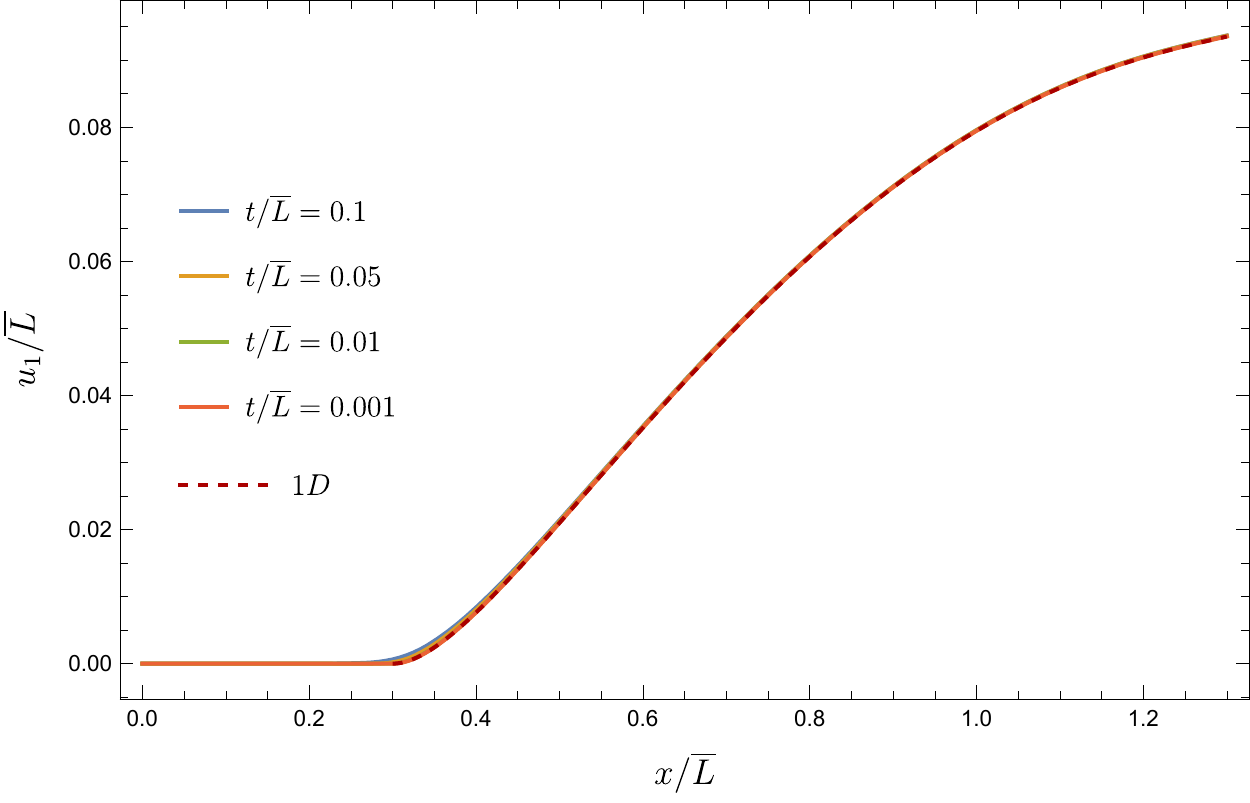}
\includegraphics[width=8cm]{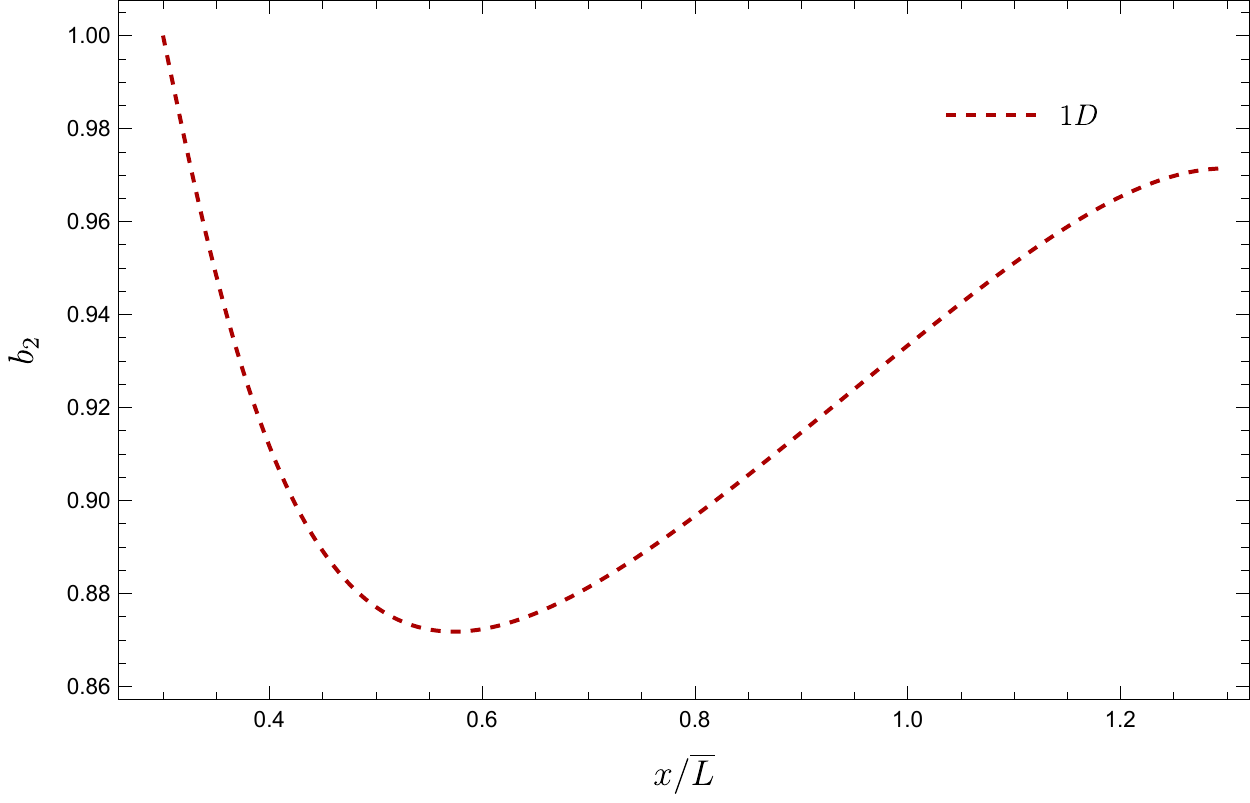}
\end{center}
\caption{Horizontal displacement of the center line (left) and director field (right) for the analysis (a) with a a distributed load $q$, without any reinforcing fibres.
\label{Fig0}}
\end{footnotesize}
\end{figure}

To assess further the numerical convergence of the model cases (b) and (c) are examined in Fig.~\ref{fig:1}. The results are shown in terms of  the deformation of the centreline $\Gamma=\lbrace(x_1,x_2)\in \Omega_t\quad \vert\quad x_2=0\rbrace$, which is plotted for $t/\overline{L}\in\lbrace 0.001,0.01,0.05,0.1\rbrace$. The insets in the figure show the deformed configuration of the 2D model for $t/\overline{L}=0.05$. The convergence of the 2D model to the 1D one is apparent in the figure and in both (b) and (C) the deformed configurations are practically indistinguishable for thickness to length--ratio lower than $0.01$, with the 1D model being represented with a dashed line.

For the sake of completeness we show in Fig.~\ref{fig:2} the director field $b$ corresponding to the solutions in Fig.~\ref{fig:1}. It is noted that in case (b) (left of Fig.~\ref{fig:2}) the vector $b$ practically does not change its length, whereas for the deformation of case (c) (right of Fig.~\ref{fig:2}), close to the clamp, the vector $b$ is compressed by 20 \%. 

\begin{figure}[H]
\begin{center}
\includegraphics[width=8cm]{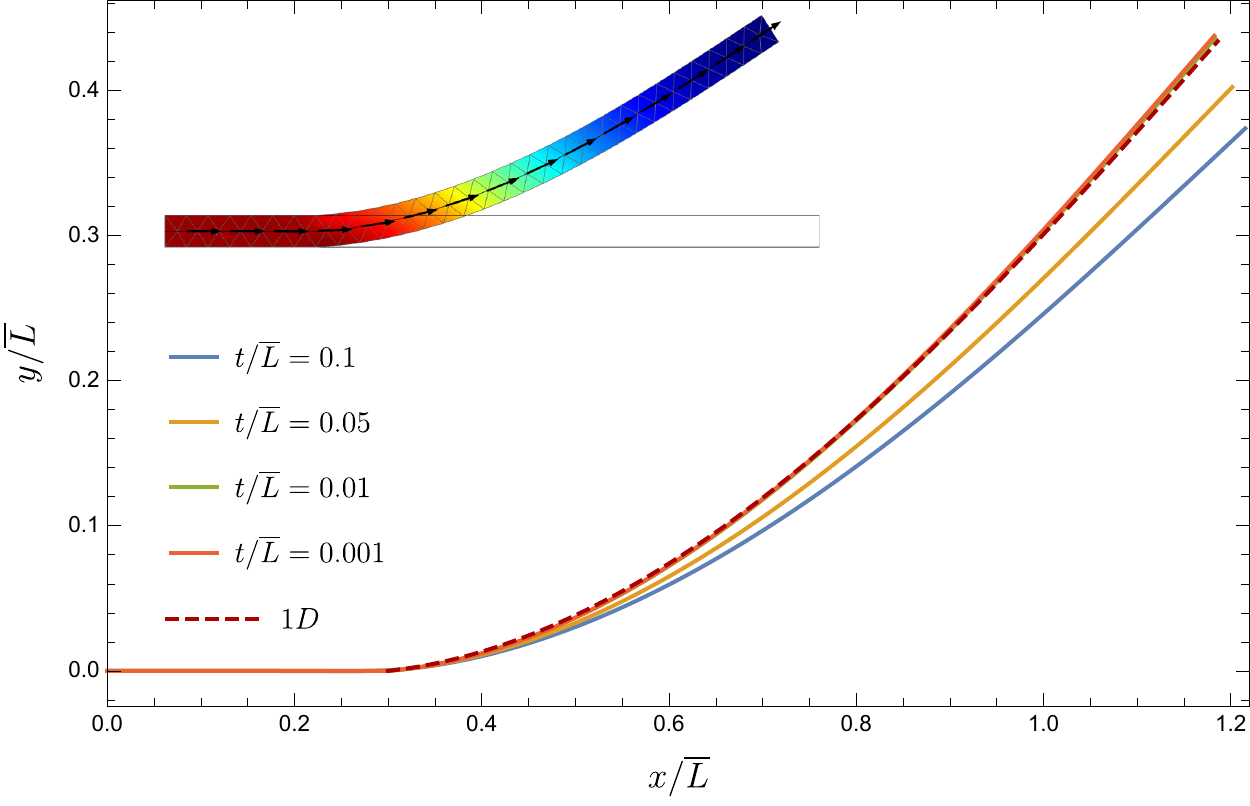}\!
\includegraphics[width=8cm]{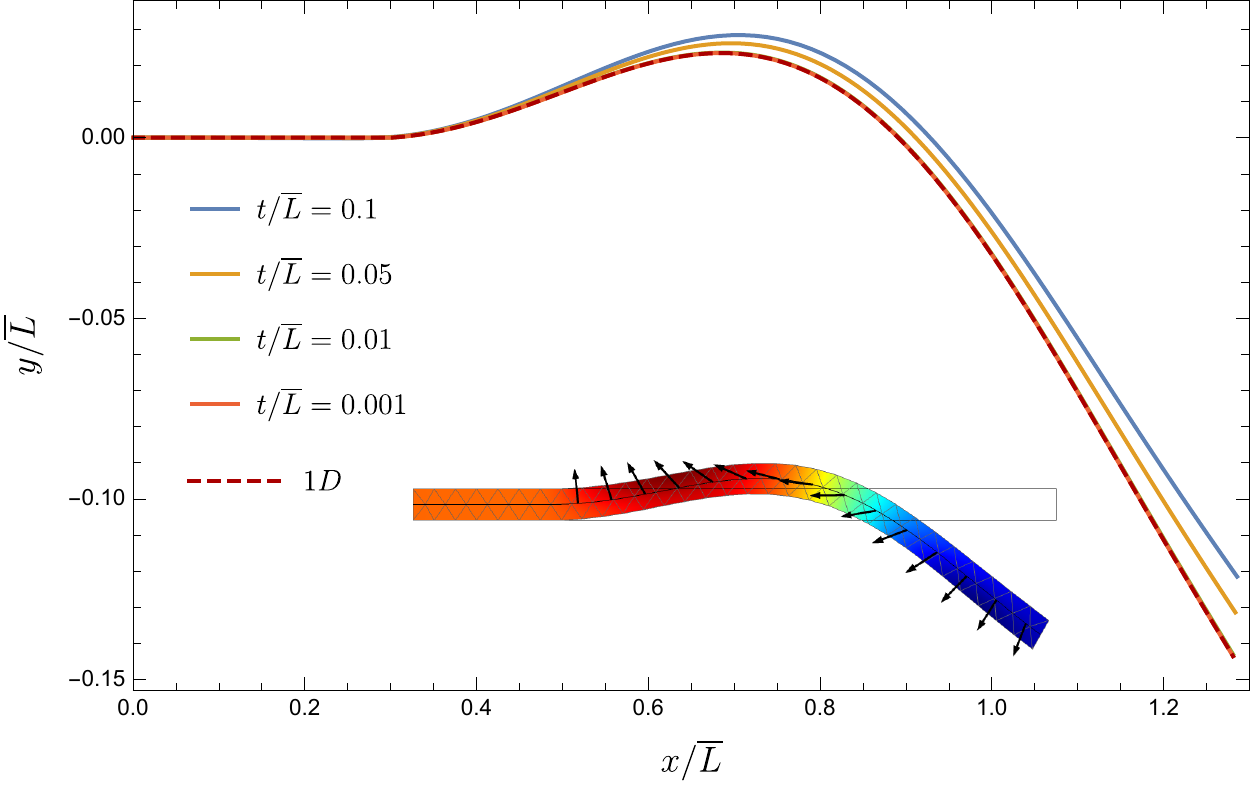}
\end{center}
\caption{Center line deformation for different values of the ratio $t/\overline{L}\in\lbrace 0.001,0.01,0.05,0.1\rbrace$ between the thickness and internal length, for case (\ref{b}) - left  ($\theta=0$ and $\varphi=\pi/3$) and (\ref{c}) - right ($\theta=\frac\pi 2 + \frac{\pi}{\overline{L}}\,(x_1-\Delta)$ and $\varphi=\pi$). The inset shows the deformed configuration for the case $t/\overline{L}=0.05$ with the color code representing the values of the higher order term $\chi_{11}$.}
\label{fig:1}
\end{figure}

\begin{figure}[H]
\begin{center}
\includegraphics[width=8cm]{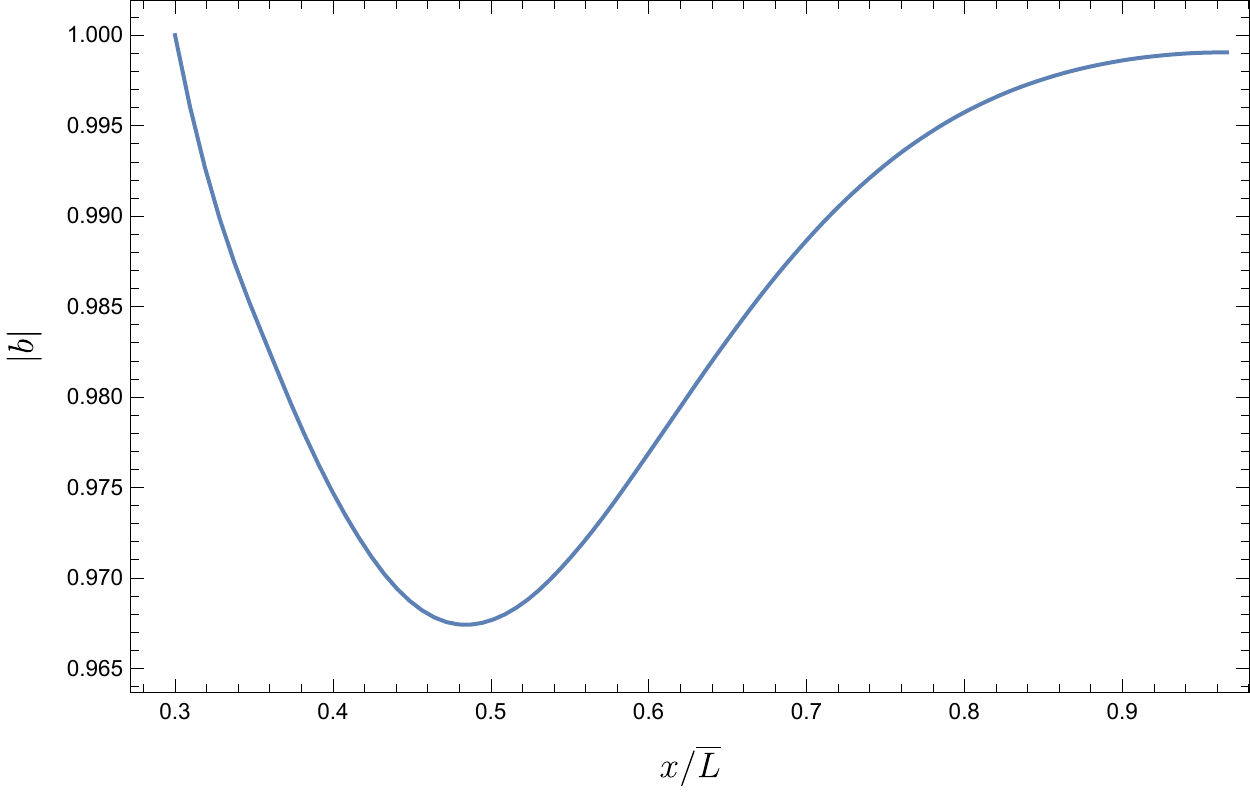}\!
\includegraphics[width=8cm]{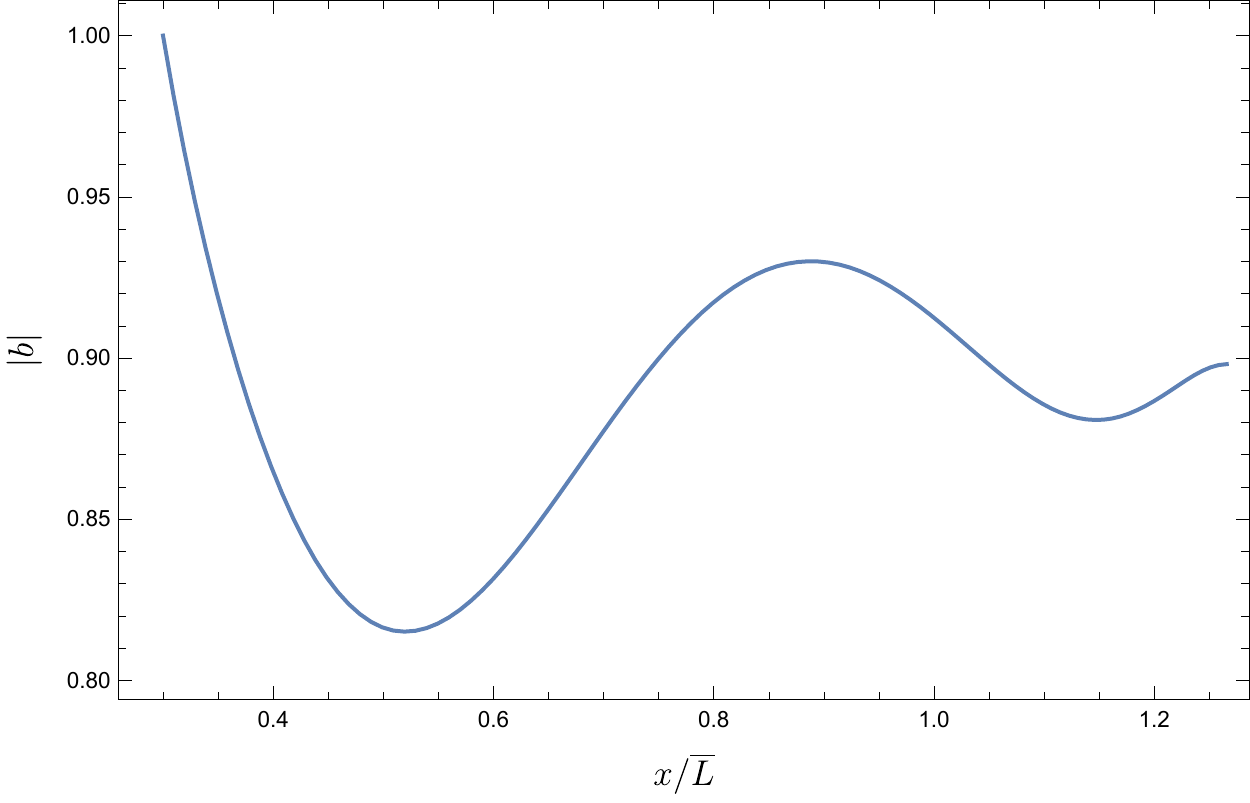}
\end{center}
\caption{Intensity of the vector field $b$ for the solutions of case (b) (left) and (c) (right) in Fig.~\ref{fig:1}.}
\label{fig:2}
\end{figure}



\section*{Appendix}

\begin{definition}[Polyconvexity in dimension 2]
We say that $\varphi: \mathbb{R}^{2 \times 2} \rightarrow \mathbb{R} \cup\{+\infty\}$ is polyconvex if there exists a convex and lower semicontinuous function $\hat{\varphi}: \mathbb{R}^{2\times 2}\times\mathbb R\rightarrow \mathbb{R} \cup\{+\infty\}$ such that $\varphi({F})=\hat{\varphi}({F},\det{F})$.
\end{definition}

\begin{proposition}There exist $a,h\in\mathbb R^2$ with $|a|=1$ such that the map $\varphi:\mathbb R^{2\times 2}\to\mathbb R$ defined by
  \begin{equation}
    \varphi(F)=\frac{h\cdot Fa}{|Fa|}
  \end{equation}
  is not rank-1 convex.
\end{proposition}
\begin{proof}
Let $b,c\in\mathbb R^2$. For $|\lambda|$ small, we have
$\operatorname{det}(I+\lambda b\otimes c)>0,$ hence the function
$$
f(\lambda)=\varphi\left(I+\lambda b\otimes c\right)
$$
is well defined and twice coutinuouly differentrable. Then
$$
\begin{aligned}
f(\lambda)=\frac{h \cdot a+\lambda(c \cdot a)(b \cdot h)}{|a+\lambda(c \cdot a) b|} &=\frac{h \cdot a+\lambda(c \cdot a)(b \cdot h)}{\sqrt{1+2 \lambda(c \cdot a)(b \cdot a)+\lambda^{2}(c\cdot a)^{2}|b|^{2}}} \\
&=\frac{\alpha+\lambda \beta}{\sqrt{1+2 \gamma \lambda+\lambda^{2} \delta^{2}}},
\end{aligned}
$$
where we have set $\alpha=h\cdot a$, $\beta=(c\cdot a)(b\cdot h)$, $\gamma=(c\cdot a)(b\cdot a)$ and $\delta=(c\cdot a)|b|$.  Note that $f$ is not necessarily convex, i.e., $\varphi$ is not rank-one convex and it implies that it is not polyconvex \cite{dacorogna}.
  \end{proof}

\section*{Acknowledgments}
GT is supported by Project PRIN 2017 \#2017KL4EF3\_004, ``Mathematics of active materials: From mechanobiology to smart devices'', and the Italian MIUR Project of Departments of Excellence. JC and GT acknowledge support from the Italian INdAM-GNFM. JC acknowledges the support of Sapienza University through the grant RM11916B7ECCFCBF.  MK acknowledges the support by the GA\v{C}R-FWF project 19-29646L. \color{black}



\begin{thebibliography}{9}


\bibitem{ball1977}  Ball, J.~M. (1976). Convexity conditions and existence theorems in nonlinear elasticity. {\em Arch. Ration. Mech. Anal.}, 63, pp. 337--403. \url{https://doi.org/10.1007/BF00279992}


\bibitem{bhatta1999} Bhattacharya, K. and James, R. D. (1999) A theory of thin films of martensitic materials withapplications to microactuators. {\it J. Mech. Phys. Solids}, 47, 531--576. \url{https://doi.org/10.1016/S0022-5096(98)00043-X}.


\bibitem{brown1966} Brown, W.F.:{\it Magnetoelastic interactions}, Springer, 1966.

  
\bibitem{Chen2017} Chen, X.-Z., Hoop, M., Mushtaq, F., Siringil, E., Hu, C., Nelson, B. J., \& Pan\'e, S. (2017). Recent developments in magnetically driven micro- and nanorobots. {\it Appl. Mat.}, 9, 37--48. \url{https://doi.org/10.1016/j.apmt.2017.04.006}


\bibitem{Chen2020} Chen, Z., Lin, Y., Zheng, G., Yang, Y., Zhang, Y., Zheng, S., Li, J., Li, J., Ren, L., \& Jiang, L. (2020). Programmable Transformation and Controllable Locomotion of Magnetoactive Soft Materials with 3D-Patterned Magnetization.
{\it ACS Appl. Mater.}, 52,
58179--58190. \url{https://doi.org/10.1021/acsami.0c15406}

\bibitem{Ciambella2017} Ciambella, J., Stanier, D. C. \& Rahatekar, S. S. (2017). Magnetic alignment of short carbon fibres in curing composites. {\it Compos. B. Eng.}, 109, 129--137. \url{http://doi.org/10.1016/j.compositesa.2016.10.001}


\bibitem{Ciambella2018} Ciambella, J., Favata, A., \& Tomassetti, G. (2018). A nonlinear theory for fibre-reinforced magneto-elastic rods. {\it Proc. Royal Soc. A}, 474, 20170703. \url{https://doi.org/10.1098/rspa.2017.0703}


\bibitem{Ciambella2019b} Ciambella J., Nardinocchi, P. (2019). Magneto-induced remodelling of fibre-reinforced elastomers, {\it Int. J. Non Linear Mech.}, 117, 103230. \url{https://dx.doi.org/10.1016/j.ijnonlinmec.2019.07.015}

\bibitem{Ciambella2020} Ciambella, J. \& Tomassetti, G. (2020). A form-finding strategy for magneto-elastic actuators. {\it Int. J. Non Linear Mech.}, 119, 103297 \url{https://doi.org/10.1016/j.ijnonlinmec.2019.103297}

\bibitem{ciarletnecas}
Ciarlet, P.G., Ne\v{c}as,  J.: Injectivity and self-contact in nonlinear elasticity. {\it Arch. Ration. Mech. Anal.} {\bf  97}
(1987),  171--188.\url{https://doi.org/10.1007/BF00250807}

\bibitem{comsol} COMSOL Multiphysics \textregistered v. 5.4. \url{www.comsol.com}, COMSOL AB, Stockholm, Sweden.

  \bibitem{dacorogna}
 Dacorogna, B.: {\it Direct methods in the calculus of variations}. Springer Business and Science Media, New York, 2008.

\bibitem{dibenedetto}
 Di Benedetto, E.: {\it Real analysis}. Birkh\"auser, Boston, 2002.


\bibitem{Danas2012} Danas, K., Kankanala, S. V \& Triantafyllidis, N. (2012). Experiments and modeling of iron-particle-filled magnetorheological elastomers. {\it J. Mech. Phys. Solids}, 60, 120--138. \url{http://doi.org/10.1016/j.jmps.2011.09.006}

\bibitem{Davolietal}
Davoli, E., Kru\v{z}\'{\i}k, M., Piovano, P.,  Stefanelli, U. (2021). Magnetoelastic thin films at large strains. {\it Contin. Mech. Thermodyn.} {\bf 33}, 327--341. \url{https://doi.org/10.1007/s00161-020-00904-1}

\bibitem{Dorfmann2003} Dorfmann, L., \& Ogden, R. W. Magnetoelastic modelling of elastomers (2003). {\it Eur. J. Mech. A Solids} {\bf 22}(4), 497-507.

\bibitem{Duan2007} Duan, W. H., Wang, C. M., \& Zhang, Y. Y. (2007). Calibration of nonlocal scaling effect parameter for free vibration of carbon nanotubes by molecular dynamics. {\it J. Appl. Phys.}, 101, 024305. \url{https://doi.org/10.1063/1.2423140}

\bibitem{durastanti} Durastanti, R., Giacomelli, L., \&  Tomassetti G. (2021). Shape Programming of a Magnetic Elastica. {\it Math. Mod. Meth. Appl. Sci.}, S0218202521500160. \url{https://doi.org/10.1142/S0218202521500160}.


\bibitem{Evans2007} Evans, B. A., Shields, A. R., Carroll, R. L., Washburn, S., Falvo, M. R., \& Superfine, R. (2007).
  Magnetically actuated nanorod arrays as biomimetic cilia. {\it Nano Letters} 5, 1428--1434. \url{https://doi.org/10.1021/nl070190c}.
  
  \bibitem{Fucik-Kufner-John}
  Fu\v{c}\'{\i}k, S., Kufner, A., John, O. (1977). {\it Function Spaces} Springer, Netherlands, 1977.
  
\bibitem{Hanasoge2020} Hanasoge, S., Hesketh, P. J., \&  Alexeev, A. (2020). Metachronal Actuation of Microscale Magnetic Artificial Cilia. {\it ACS Appl. Mater. Interfaces}, 41, 46963--46971. \url{https://doi.org/10.1021/acsami.0c13102},

\bibitem{healeykroemer}
Healey, T.J.,  Kr\"{o}mer, S. (2009). Injective weak solutions in second-gradient
nonlinear elasticity. {\it ESAIM Control Optim. Calc. Var.} {\bf  15},  863--871. \url{https://doi.org/10.1051/cocv:2008050}

\bibitem{Horak2020}  Hor\'ak, M., Kruz\'ik, M.  (2020). Gradient polyconvex material models and their numerical treatment. {\it Int. J. Solids Struct.}, 195, 57--65. \url{https://doi.org/10.1016/j.ijsolstr.2020.03.006}

  \bibitem{Hu2018} Hu, W., Lum, G. Z., Mastrangeli, M., \& Sitti, M. (2018). Small-scale soft-bodied robot with multimodal locomotion. {\it Nature}, 554, 81--85. \url{https://doi.org/10.1038/nature25443}

\bibitem{Kankanala2004} Kankanala, S. V. \& Triantafyllidis, N. (2004). On finitely strained magnetorheological elastomers. {\it J. Mech. Phys. Solids}, 52, 2869--2908. \url{https://doi.org/10.1016/j.jmps.2004.04.007}.

  \bibitem{kruzikroubicek}
Kru\v z\'\i k, M.,  Roub{\'i}{\v c}ek, T.: {\it Mathematical Methods in
  Continuum Mechanics of Solids}. Springer Nature Switzerland AG, Cham, 2019.

\bibitem{Lum2016} Lum, G. Z., Ye, Z., Dong, X., Marvi, H., Erin, O., Hu, W., \& Sitti,  M. (2016). Shape-programmable magnetic soft matter. {\it Proc. Nat. Acad. Sci.}, 113,   E6007--E6015. \url{https://doi.org/10.1073/pnas.1608193113}


\bibitem{Li2015} Li, C., Li, S. Yao, L. \& Zhu Z. (2015) Nonlocal theoretical approaches and atomistic simulations for longitudinal free vibration of nanorods/nanotubes and verification of different nonlocal models. {\it Appl. Math. Model.}, 39, 4570--4585. \url{https://doi.org/10.1016/j.apm.2015.01.013}

\bibitem{Mathematica} {Mathematica, {V}ersion 12.3}, \url{https://www.wolfram.com/mathematica}, {Wolfram Research{,} Inc.}, {Champaign, IL, 2021}
  
  \bibitem{Murmu2014} Murmu, T., Adhikari, S., \& McCarthy, M. A. (2014). Axial vibration of embedded nanorods under transverse magnetic field effects via nonlocal elastic continuum theory. {\it J. Comput. Theor. Nanosci.}, 5, 1230--1236. \url{https://10.1166/jctn.2014.3487}


\bibitem{Rikken2014} Rikken, R. S. M., Nolte, R. J. M., Maan, J. C., Van Hest, J. C. M., Wilson, D. A., \& Christianen, P. C. M. (2014). Manipulation of micro- and nanostructure motion with magnetic fields. {\it Soft Matter}, 10, 1295--1308. \url{https://doi.org/10.1039/c3sm52294f}

\bibitem{sano2021} Sano, T. G., Pezzulla, M., and Reis, P. (2021). A Kirchhoff-like theory for hard magnetic rods under geometrically nonlinear deformation in three dimensions. \texttt{arXiv:2106.15189}. \url{https://arxiv.org/abs/2106.15189}

\bibitem{schraad1997} Schraad, M. W. and Triantafyllidis, N. (1997) Scale Effects in Media With Periodic and Nearly Periodic Microstructures, Part I: Macroscopic Properties. {\it J. Appl. Mech.}, 64:751--762. \url{https://doi.org/10.1115/1.2788979}.


\bibitem{Stanier2016} Stanier, D. C., Ciambella, J., \& Rahatekar, S. S. (2016). Fabrication and characterisation of short fibre reinforced elastomer composites for bending and twisting magnetic actuation. {\it Compos. Part A Appl. Sci. Manuf.}, 91, 168-176. \url{https://doi.org/10.1016/j.compositesa.2016.10.001}



\bibitem{Tiersten1964} Tiersten, H. F. (1964). Coupled magnetomechanical equations for magnetically saturated insulators. {\it J. Math. Physics}, 5, 1298–1318. \url{https://doi.org/10.1063/1.1704239}

\bibitem{Wang2020} Wang, L., Kim, Y., Guo, C. F. \& Zhao, X (2020). Hard-magnetic elastica. {\it J. Mech. Phys. Solids}, 142, 104045. \url{https://doi.org/10.1016/j.jmps.2020.104045}


\bibitem{Weiyi1999} Weiyi, C. (1999). Derivation of the general form of elasticity tensor of the transverse isotropic material by tensor derivate. {\it Appl. Math. Mech.}, 20, 309--314. \url{https://doi.org/10.1007/BF02463857}.

\bibitem{Wu2020} Wu, S., Hu, W., Ze, Q., Sitti, M. \& Zhao, R. (2020). Multifunctional magnetic soft composites: A review. {\it Multifunct. Mater.}, 3, 042003. \url{https://doi.org/10.1088/2399-7532/abcb0c}

  \bibitem{Xu2019} Xu, T., Zhang, J., Salehizadeh, M., Onaizah, O. \& Diller, E. (2019) Millimeter-scale flexible robots with programmable three-dimensional magnetization and motions. Sci. Robot. 4, eaav4494. \url{https://doi.org/10.1126/scirobotics.aav4494}

\bibitem{Ze2020} Ze, Q., Kuang, X., Wu, S., Wong, J., Montgomery, S. M.,  Zhang, R. and Kovitz, J. M., Yang,  F., Qi, H. J., Zhao, R. (2020). Magnetic Shape Memory Polymers with Integrated Multifunctional Shape Manipulation. {\it Adv. Mater.}, 32, 1906657. \url{https://doi.org/10.1002/adma.201906657}

\end{thebibliography}
\end{document}